\pdfoutput=1

\documentclass[10pt]{article}

\usepackage[numbers,sort&compress]{natbib}
\usepackage{amsthm}
\usepackage{amssymb}
\usepackage{hyperref}
\usepackage{graphicx}
\usepackage{xspace}
\usepackage{url}
\usepackage{subcaption}
\usepackage{enumitem}
\usepackage{authblk}

\graphicspath{{./pictures/}}

\usepackage[T1]{fontenc}
\usepackage[tt=false]{libertine}
\usepackage[libertine]{newtxmath}
\usepackage{inconsolata}
\usepackage{geometry}
\usepackage{mathtools}

\newtheorem{theorem}{Theorem}
\newtheorem{lemma}{Lemma}

\usepackage{booktabs}

\title{Universal Wait-Free Memory Reclamation}

\author{Ruslan Nikolaev, Binoy Ravindran}
\affil{rnikola@vt.edu, binoy@vt.edu\\Virginia Tech, USA}

\date{}

\begin{document}

\maketitle

\begin{abstract}
In this paper, we present a universal memory reclamation scheme, \textit{Wait-Free Eras} (WFE), for deleted memory blocks in wait-free concurrent data structures. WFE's key innovation is that it is completely wait-free. Although some prior techniques provide similar guarantees for certain data structures, they lack support for \textit{arbitrary} wait-free data structures.
Consequently, developers are typically
forced to marry their wait-free data structures with lock-free
Hazard Pointers or (potentially blocking) epoch-based memory 
reclamation. Since both these schemes provide weaker progress guarantees, they essentially forfeit the strong progress guarantee of wait-free data structures. Though making the original Hazard Pointers scheme or epoch-based reclamation completely wait-free seems infeasible, we achieved this goal with a more recent, (lock-free) Hazard Eras scheme, which we extend to guarantee wait-freedom. As this extension is non-trivial, we discuss all challenges pertaining to the construction of universal wait-free memory reclamation.

WFE is implementable
on ubiquitous x86\_64 and AArch64 (ARM) architectures.
Its API is mostly compatible with
Hazard Pointers, which allows easy transitioning of existing
data structures into WFE. Our experimental evaluations show that WFE's performance is close to epoch-based reclamation and almost matches the original Hazard Eras scheme, while providing the  stronger wait-free progress guarantee.

\end{abstract}

\def\keywords{\vspace{.5em}
{\noindent{\textit{Keywords}:\,\relax%
}}}
\def\endkeywords{\par}

\keywords{wait-free, non-blocking, memory reclamation, hazard pointers, hazard eras}

\section{Introduction}

Most modern general purpose systems use the shared-memo\-ry architecture, which stipulates solving the synchronization problem when accessing shared data. The easiest way to solve this problem is to use locks, but due to scalability issues, non-blocking data structures have been studied over the years. A downside of non-blocking data structures is that they require a special \textit{memory reclamation} scheme. Whereas mutual exclusion locks can guarantee that no other thread is using a memory block (node) that is in the process of deletion, this is generally not true for non-blocking designs.

\textit{Wait-freedom}, the strongest of non-blocking progress guarantees, is critically important in many latency-sensitive applications where execution time of all operations must be  bounded~\cite{Herlihy:2008:AMP:1734069}. In wait-free algorithms, all threads must eventually complete any operation after a bounded number of steps. Nonetheless, such algorithms have not had significant practical traction due to a number of reasons. Traditionally, wait-free algorithms were difficult to design and were much slower than their lock-free counterparts. Kogan and Petrank's  \textit{fast-path-slow-path} methodology~\cite{Kogan:2012:MCF:2145816.2145835} largely solved the problem of creating efficient wait-free algorithms. The design of wait-free algorithms, however, is still challenging, 
as the Kogan-Petrank methodology implicitly assumes wait-free memory reclamation. A case in point:  Yang and Mellor-Crummey's 
wait-free queue~\cite{ChaoranWFQ} uses the Kogan-Petrank  methodology. But, as pointed out by Ramalhete and Correia~\cite{pedroWFQUEUEFULL},~\cite{ChaoranWFQ}'s design is flawed in its memory reclamation approach which, strictly speaking, forfeits wait-freedom.

Although a number of memory reclamation  techniques~\cite{epoch1,epoch2,HPPaper,ThreadScan,ForkScan,Balmau:2016:FRM:2935764.2935790,DEBRAPaper,refcount2,refcount3,IBRPaper,Hyaline,HEPaper,Cohen:2015:AMR:2814270.2814298,Cohen:2018:DSD:3288538.3276513,Braginsky:2013:DAL:2486159.2486184,Herlihy:2002:ROP:645959.676129,Herlihy:2005:NMM:1062247.1062249} have been proposed, only a fraction of them can be used for arbitrary data structures and are truly non-blocking~\cite{HEPaper,HPPaper,IBRPaper,Hyaline,refcount2,refcount3,Cohen:2018:DSD:3288538.3276513,Herlihy:2002:ROP:645959.676129,Herlihy:2005:NMM:1062247.1062249}. At
present, no \textit{universal} memory reclamation technique exists that guarantees wait-freedom for \textit{arbitrary} wait-free data structures. Typically, prior efforts on wait-free data structures have ignored the memory reclamation
problem entirely or have harnessed lock-free memory reclamation schemes such as Hazard Pointers~\cite{HPPaper}, essentially forfeiting strict wait-freedom guarantees. The recently proposed Ramalhete and Correia's wait-free queue~\cite{pedroWFQUEUE} can be implemented using Hazard Pointers, but the approach is too specific to the queue's design and cannot be applied for other data structures.~\cite{OneFile} presents wait-free memory reclamation for software transactional memory (STM), but the work does not consider wait-free memory reclamation for handcrafted data structures. Although STM is an important general synchronization technique, handcrafted data structures usually perform better and still need wait-free memory reclamation.

We present a universal scheme, \textit{Wait-Free Eras} (WFE), that
solves the wait-free memory reclamation problem for \textit{arbitrary} data structures. For the first time, wait-free data structures such as the Kogan-Petrank queue~\cite{kpWFQUEUE} can be implemented fully wait-free using WFE.
Additionally, WFE can help simplify memory reclamation for existing wait-free
data structures which use Hazard Pointers in a special way to guarantee
wait-freedom.

WFE is based on the recent Hazard Eras~\cite{HEPaper} memory reclamation approach, which is lock-free. We demonstrate how Hazard Eras can be extended to guarantee wait-freedom.
In our evaluation, we observe that WFE's performance is close to epoch-based reclamation and almost matches the original Hazard Eras scheme.

\begin{figure*}
\begin{subfigure}{.49\textwidth}
\includegraphics[width=\columnwidth]{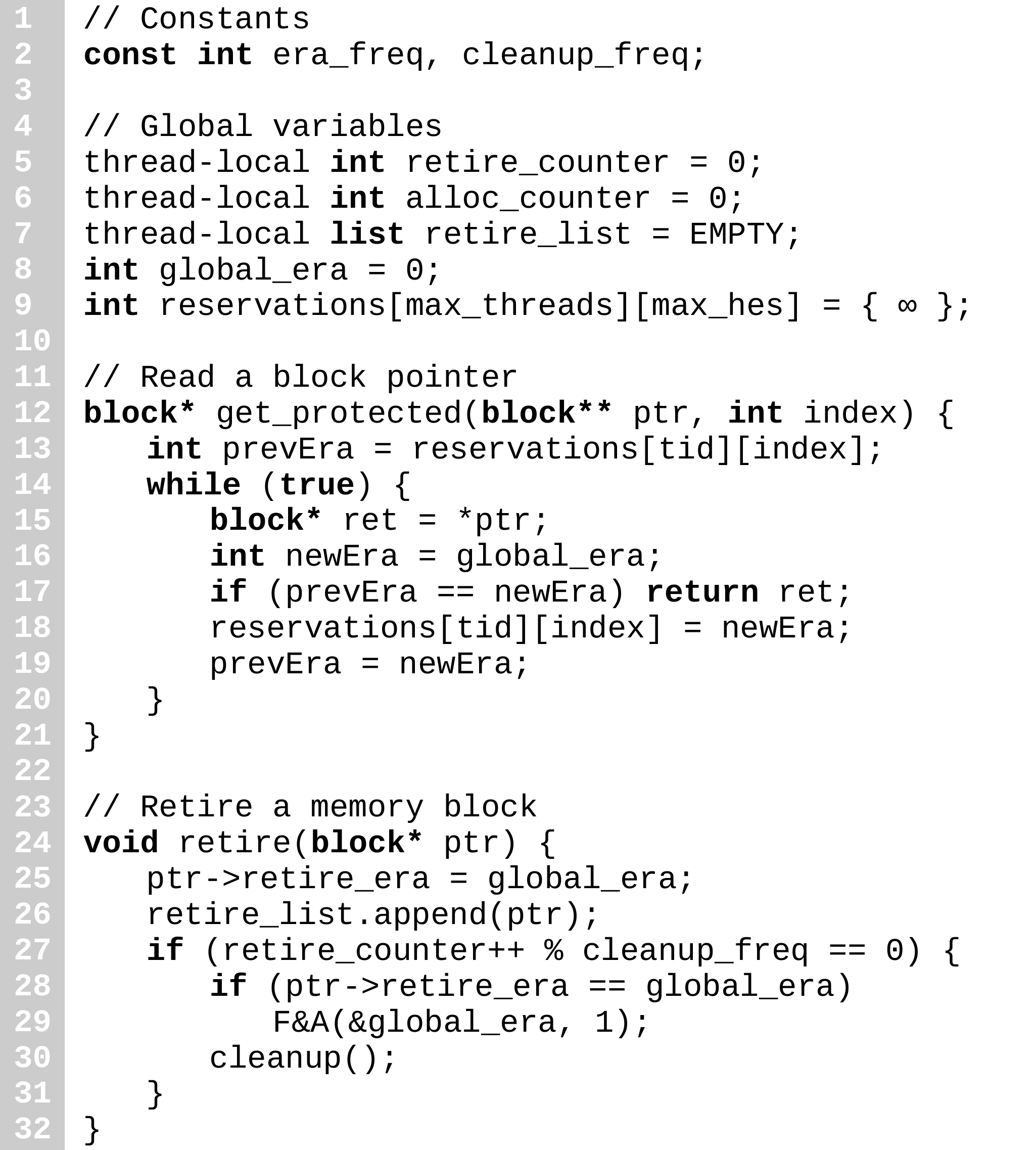}
\end{subfigure}
\begin{subfigure}{.49\textwidth}
\includegraphics[width=\columnwidth]{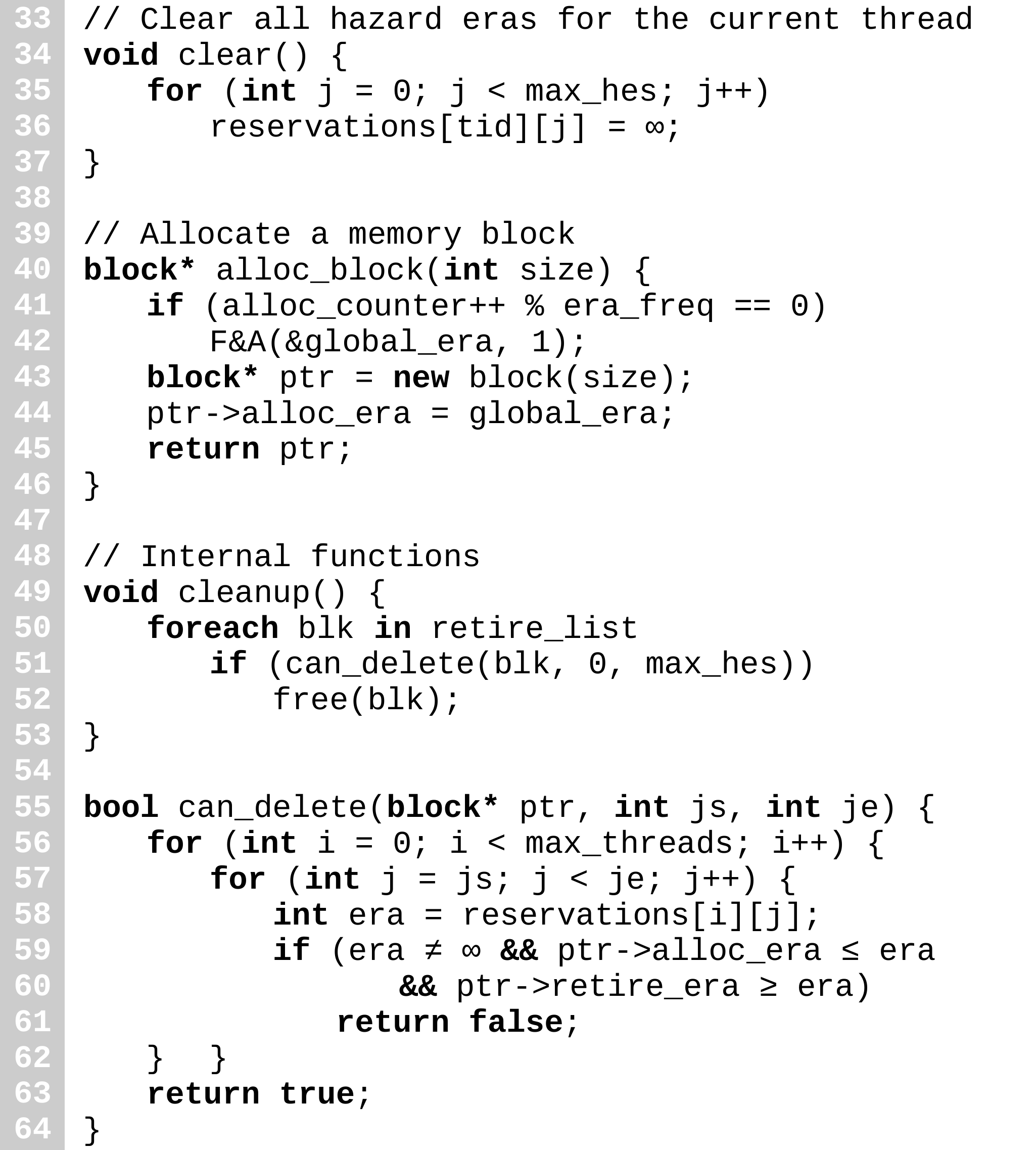}
\end{subfigure}
\caption{The Hazard Eras memory reclamation scheme.}
\label{fig:he}
\end{figure*}

\section{Background}

For greater clarity and completeness, we discuss
relevant memory reclamation schemes and the challenges in designing them with wait-free progress guarantees.

\subsection{Progress Guarantees} 

Non-blocking data structures can provide different
\textit{progress guarantees}. In \textit{obstruction-free} algorithms,
a thread performs an operation in a finite number of steps
if executed in isolation from other threads. In \textit{lock-free}
algorithms, at least one thread always makes progress in a finite number
of steps. Finally, \textit{wait-freedom} -- the strongest progress guarantee --
implies that \textit{all} threads make progress in a finite number of
steps. Wait-free data structures are particularly useful for latency-sensitive applications which usually have quality of service constraints.

Memory reclamation adds extra requirements to progress guarantees.
Unless memory usage is bounded, threads will be unable to allocate
memory at some point. This effectively blocks threads from
making further progress. Thus, memory reclamation schemes must 
guarantee that stalled or preempted threads will not prevent timely
memory reclamation.

The epoch-based reclamation (EBR) scheme~\cite{epoch1,epoch2} can have unbounded memory usage, preventing its use in wait-free algorithms. In contrast, reclamation schemes such as Hazard Pointers~\cite{HPPaper} and Hazard Eras~\cite{HEPaper} provide strict memory bounds as long as programs properly use these  schemes. However, both schemes lack \textit{wait-free} progress guarantees.

\subsection{Atomic Operations}
Typically, lock-free and wait-free data structures are implemented using compare-and-swap (CAS) operations. CAS is more general than 
other atomic operations such as fetch-and-add (F\&A), which can be
emulated by CAS. Nonetheless, x86\_64 and 
AArch64 (ARM; as of version 8.1) implement F\&A natively due to its better efficiency in hardware. Moreover,
hardware's F\&A execution time is bounded, which makes it appealing for wait-free algorithms.

In this paper, we also use wide CAS (WCAS), which  updates
two \textit{adjacent} memory words. WCAS is not to be confused
with double-CAS, which updates two \textit{arbitrary} words but is rarely
supported in commodity hardware. WCAS, however, is available in x86\_64 and 
AArch64 architectures. Furthermore, WCAS is required by commodity OSes such as Windows 8.1 or higher~\cite{CMPXCHG16BWINDOWS}.

\subsection{Hazard Eras}

The Hazard Eras~\cite{HEPaper} memory reclamation scheme merges EBR  
with Hazard Pointers.
Each allocated object retains two fields used by the reclamation scheme: \emph{``alloc\_era''} and \emph{``retire\_era''}.
When allocating a new object, its \emph{alloc\_era} is initialized
with the global era clock value, which is periodically incremented.
When the object is retired, its \emph{retire\_era} is also
initialized with the global era value. The
lifespan of the object is determined by these two eras.
When a thread accesses a hazardous reference, it publishes
the current era. Thus, if the lifespan of an object falls within \textit{any}
of the published eras, it will not be reclaimed.

Hazard Eras' API is mostly compatible with that of Hazard Pointers and consists of
the following operations:
\begin{itemize}
    \item \emph{get\_protected()}: safely retrieve a pointer to the protected object
    by creating a reservation; each object
    needs an \textit{index} that identifies
    the reservation.
    \item \emph{retire()}: mark an object for
    deletion; the retired object must be
    deleted from the data structure first, i.e.,
    only in-flight threads can still access it.
    \item \emph{clear()}: reset all prior
    reservations made by the current thread in \emph{get\_protected()}.
    \item \emph{alloc\_block()}: a special
    operation unique to Hazard Eras; it allocates
    a memory block and initializes its \emph{alloc\_era} to the global era clock value.
\end{itemize}

Figure~\ref{fig:he} presents the Hazard Eras
algorithm. In the algorithm, we assume
that the maximum number of threads is \emph{max\_threads}. The maximum number of
reservations per each thread is \emph{max\_hes}.
All retired nodes are appended to the thread-local
\emph{retire\_list}. The algorithm periodically scans this list
to check if old nodes can be safely de-allocated by calling \emph{cleanup()}.
To guarantee that the memory usage is bounded, the algorithm
periodically increments the global era
clock in \emph{alloc\_block()} and \emph{retire()}. Arguments on correctness and bounded memory usage can be found in~\cite{pedroHEFULL}.

In Figure~\ref{fig:usage}, we present an example of Treiber's stack~\cite{Treiber:tech86} implementation using Hazard Eras. The stack is a linked-list
of nodes which store pointers to inserted objects. Each node also
encapsulates a memory reclamation header block. When enqueuing an
object, we allocate a node using \emph{alloc\_block()}, store a pointer to the object, and update the stack pointer. When dequeuing, we dereference the top of the stack
using \emph{get\_protected()}. We only use \emph{index 0}
since we dereference just one pointer at a time.
The top of the stack is then updated
to refer to the next node. We retrieve an object pointer and retire the
dequeued node. Finally, \emph{clear()} resets all reservations
(i.e.,~\emph{index 0}) for the current thread. 

\subsection{Challenges in Wait-Free Memory Reclamation}

Only a few memory reclamation schemes are truly non-blocking, i.e., do
not use any OS mechanisms and also guarantee bounded memory usage.
Although certain OS mechanisms, such as signals, can transparently
handle reclamation, it is difficult to guarantee non-blocking behavior
in general, as locks are often used inside OS kernels when
dispatching signals.
We also focus on manual memory reclamation techniques, i.e., garbage
collectors are beyond the scope of this paper.
Though EBR is almost wait-free for most operations, it is blocking
due to its potentially unbounded memory usage. Since truly wait-free
algorithms must also be non-blocking, we only consider those algorithms.
We narrow our scope to Hazard Pointers~\cite{HPPaper},
Hazard Eras~\cite{HEPaper}, and Interval-Based Reclamation (IBR)~\cite{IBRPaper}.

In both Hazard Pointers and Hazard Eras, most
operations are already wait-free. However, access to hazardous
references through \emph{get\_protected()} is a noticeable exception. Both schemes need
unbounded loops:
Hazard Eras check that the published global era value has not
changed while reading the reference; Hazard Pointers publish 
the reference itself but still need to validate that the reference
has not changed since it was published. Despite this similarity, the  two schemes
differ drastically. To make Hazard Pointers wait-free, we must ensure
that the references do not change, which seems generally impossible to do
in a wait-free manner. In Hazard Eras, we only need to make sure that
the \textit{global era} value is unchanged. Although the original Hazard
Eras approach could not solve this problem, we demonstrate a viable solution.

We also considered IBR, especially because of its simpler API.
However, we preferred to extend the Hazard Eras scheme instead due to its strict memory usage guarantees,
even in the presence of starving or slow threads. (Wait-free memory reclamation
can still be used for ordinary, lock-free data structures.)
IBR requires additional changes to data structures to guard against this
condition~\cite{IBRPaper}. 
Certain tagged versions of IBR also
require more work to make them wait-free. Nonetheless, our approach is
applicable to the 2GEIBR version where only hazardous reference accesses need to be made wait-free.

\begin{figure}
\centering
\includegraphics[width=.5\columnwidth]{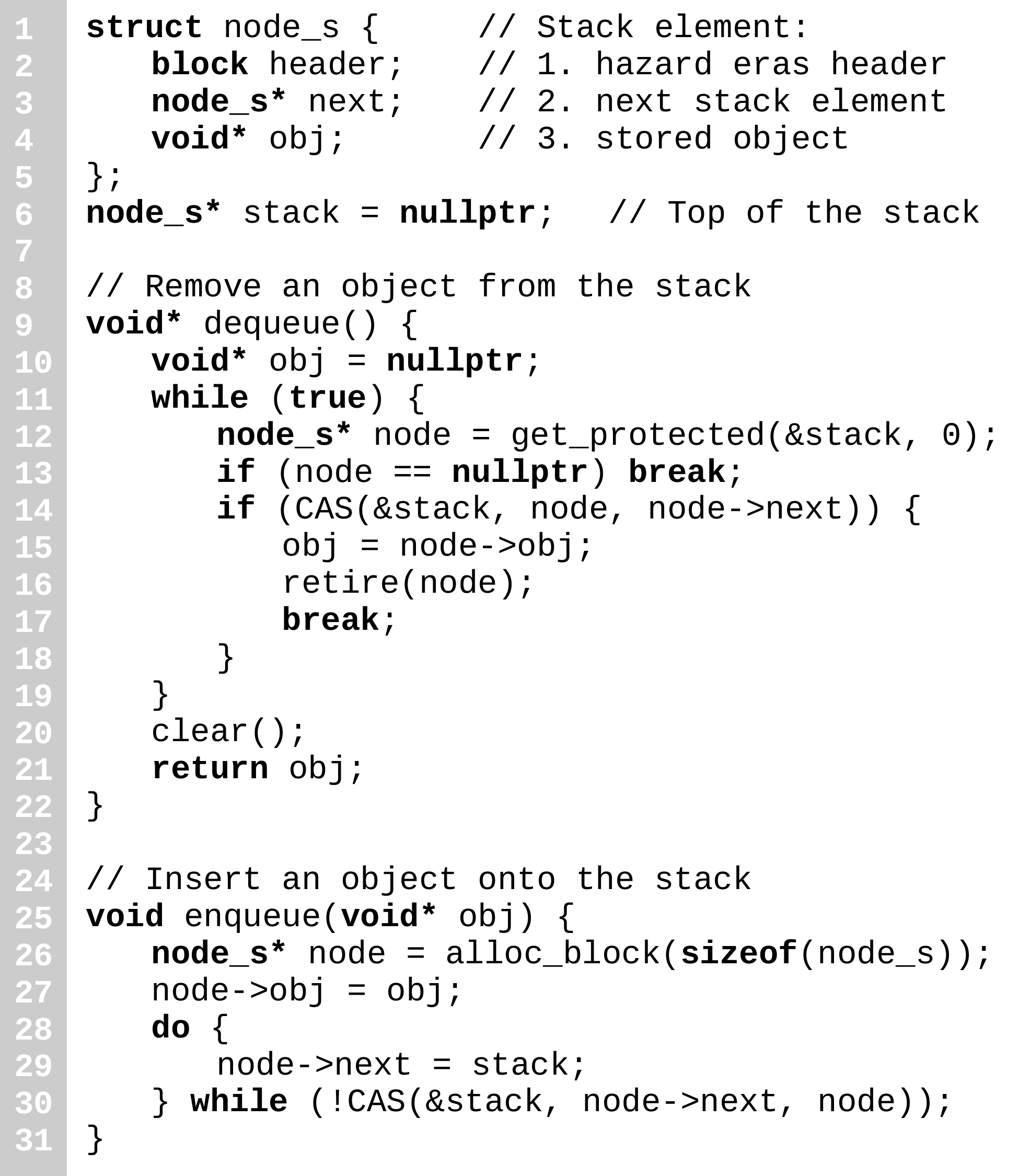}
\caption{Usage example: Treiber's lock-free stack.}
\label{fig:usage}
\end{figure}

\section{Wait-Free Memory Reclamation}

The proposed WFE algorithm employs the traditional \textit{fast-path-slow-path}
idea but avoids the Kogan-Petrank methodology~\cite{Kogan:2012:MCF:2145816.2145835} 
due to the 
complexity of memory reclamation and uniqueness of its challenges.
The high-level idea is to execute an almost unchanged Hazard Eras scheme
on the fast path. When \emph{get\_protected()} fails to complete
after a specified number of steps, the slow path is taken where threads
collaborate to guarantee wait-freedom.

\subsection{Assumptions}
We assume that the hardware supports WCAS and wait-free F\&A.
We also assume a 64-bit CPU since we need to atomically
manipulate eras (sometimes using WCAS) which, for safety, are typically at
least 64 bits wide in Hazard Eras. Our assumptions are true for common-place x86\_64
and more recent versions of AArch64 (ARM) processors.

Although not all other microarchitectures currently satisfy these requirements,
they can easily fall back to the original Hazard Eras algorithm to
retain API compatibility at the cost of losing wait-freedom.

To simplify our pseudo-code, we will further assume that
the memory model is
sequentially consistent~\cite{Lamport:1979:MMC:1311099.1311750}.
Our actual implementation relaxes this requirement to benefit
architectures with weaker memory models such as AArch64.

Finally, we assume that the number of threads is bounded, which is a reasonable
assumption made by most reclamation schemes.

\subsection{Data Fields and Formats}

We inherit most existing data fields from Hazard Eras. However, we
modify the \emph{reservations} array to record \textit{pairs} instead
of eras.
Each pair consists of an \emph{era} as well as a \emph{tag} for the
current reservation. The tag is only accessed on the slow
path and identifies the slow-path cycle. The tag is monotonically increased
after each such cycle and protects against spurious
(delayed) data changes happening afterwards.

We also reserve a special pointer value, \textbf{invptr}, which should
never be used by data structures. Since \textbf{nullptr}
is still occasionally used, we reserve the maximum integer value
instead; it should never be stored in valid pointers (e.g.,
mmap(2) returns the same value for MAP\_FAILED).

\begin{figure}
\centering
\includegraphics[width=.4\columnwidth]{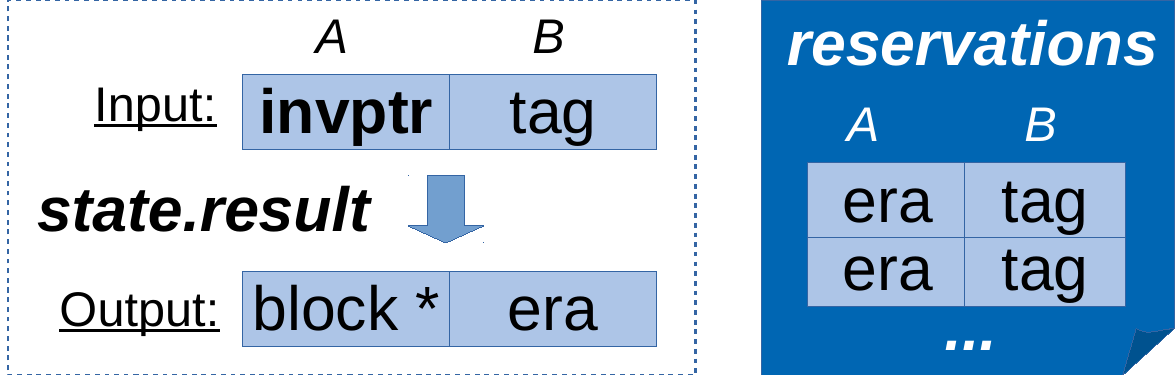}
\caption{WFE state and reservations.}
\label{fig:wfe}
\end{figure}

For the slow path, each reservation keeps \emph{state} (in a separate
array). We discuss all fields of \emph{state} below. One
of its fields, \emph{result}, is used for both input and output. On the
input, it records the \emph{tag} of the current slow-path cycle.
On the output, it contains a dereferenced pointer as well as
the era that needs to be set in \emph{reservations} for this
pointer. To distinguish these two cases, we place \textbf{invptr} in the place
of the dereferenced pointer and \emph{tag} in the place of \emph{era}.
The pointer value distinguishes whether the result is already produced.

We summarize new data fields in Figure~\ref{fig:wfe}.

\begin{figure*}
\vspace{+10pt}
\begin{subfigure}{.49\textwidth}
\includegraphics[width=\columnwidth]{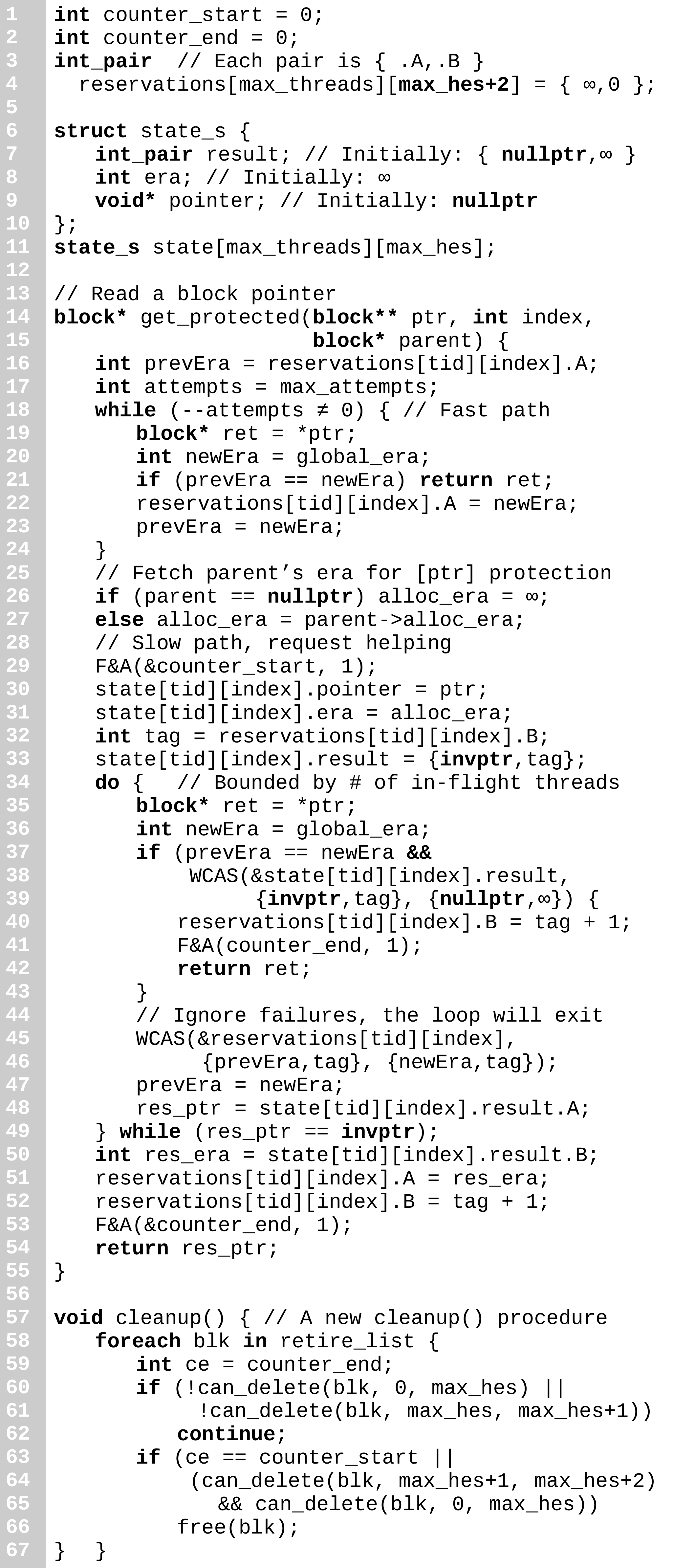}
\end{subfigure}
\begin{subfigure}{.49\textwidth}
\includegraphics[width=\columnwidth]{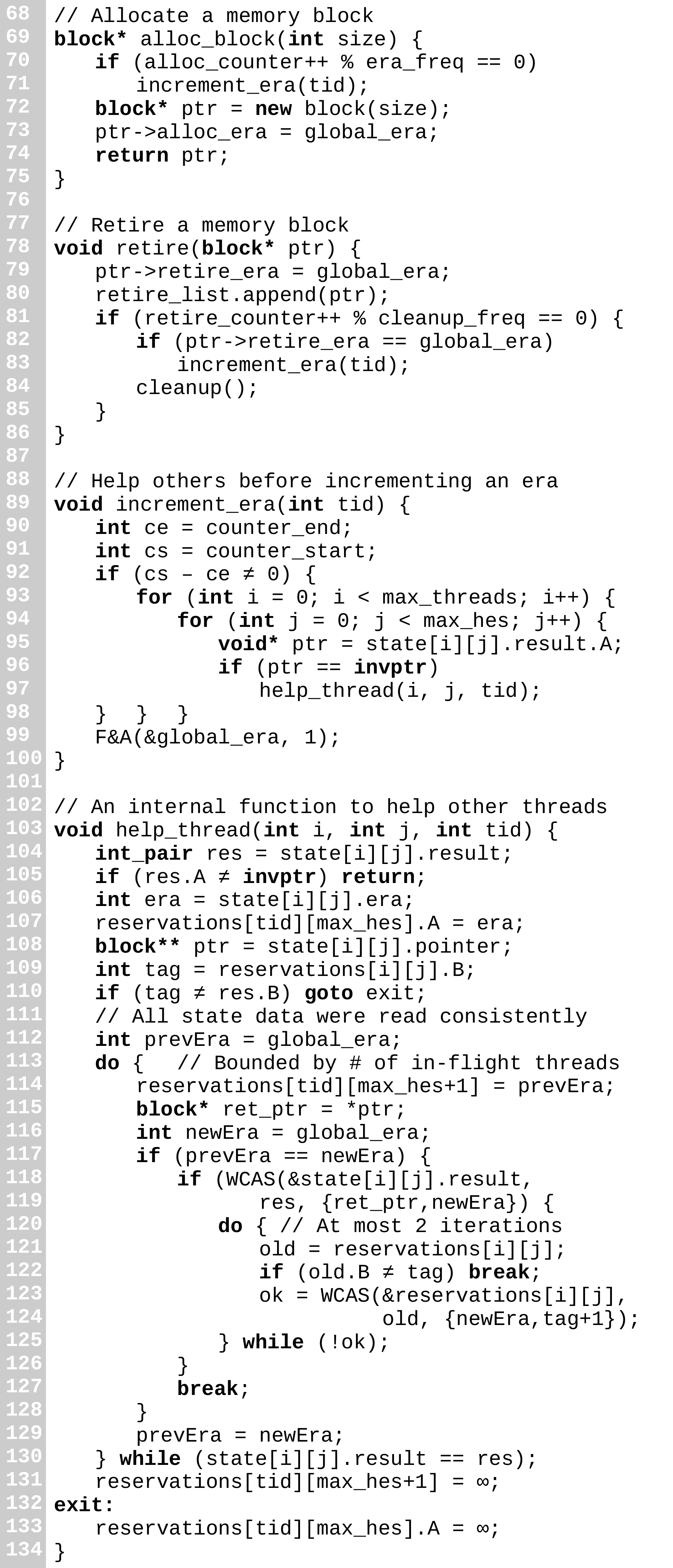}
\end{subfigure}
\vspace{+10pt}
\caption{The Wait-Free Eras (WFE) memory reclamation scheme.}
\label{fig:wfe_alg}
\end{figure*}

\subsection{Bird's-Eye View}
Assuming the presence of wait-free F\&A, as discussed in~\cite{pedroHEFULL},
Hazard Eras is already wait-free for all operations
except \emph{get\_protected()}, which is used to dereference pointers.
The \emph{get\_protected()} operation contains a potentially unbounded loop
and, consequently, is only lock-free. The only reason why
this loop may never converge is due to the changing value of the
global era clock. The global era value changes periodically when
allocating new objects or retiring old objects. Thus, to make the Hazard Eras algorithm
wait-free, threads that call \emph{alloc\_block()} and \emph{retire()} have
to collaborate with threads that call \emph{get\_protected()}.

In other words, \emph{alloc\_block()} and \emph{retire()} should not be
incrementing the global era clock
unless they can guarantee that all other running threads succeed in their \emph{get\_protected()} call. However,
global era increments cannot be simply postponed since they are crucial to
guarantee that the memory usage is bounded. Likewise, \emph{alloc\_block()}
or \emph{retire()} cannot block on \emph{get\_protected()} either. The idea
employed in WFE is that before incrementing the global era,
\emph{alloc\_block()} and \emph{retire()} will first check
if \textit{any} thread needs helping.

\subsection{Wait-Free Eras (WFE)}

Figure~\ref{fig:wfe_alg} presents WFE's pseudo-code for functions that diverge
from Hazard Eras. All other operations that use the \emph{reservations}
array need to be modified accordingly to access the \textit{A} component of the pair which
retains an era. The \emph{reservations} array is also extended by two additional
reservations per each thread, i.e., its size is now \emph{max\_hes+2}.
These two new reservations should never be used by the application; they are internal to \emph{help\_thread()}.

WFE slightly alters the API for \emph{get\_protected()}. Its extra argument, \emph{parent}, needs to provide the area (block) where
the hazardous reference is located. Typical data structures keep all references
inside their blocks (e.g., a linked-list node keeps the reference to the next
node) and the parent parameter simply
refers to the previously retrieved hazardous reference. Since the topmost
references do not have any parent, we also allow to pass \textbf{nullptr}
for them.

Lines~16-24 of \emph{get\_protected()} represent the fast path and are identical
to the original implementation. Lines~26-27 retrieve the \emph{alloc\_era}
from the parent block (if any), so that the block can be saved from
being reclaimed while retrieving the hazardous reference from it
by a helper method.

To facilitate fast detection if any thread needs helping, we maintain
global variables \emph{counter\_start} and \emph{counter\_end},  which
are atomically incremented using F\&A. Their purpose is twofold. First, the
difference between these two variables indicate the number of threads
that need helping. Second, if the value of \emph{counter\_start} changes,
it indicates that some new thread entered the slow path region.

On the slow path, a thread initializes \emph{state} with a hazardous
reference pointer and the \emph{alloc\_era} from the parent block where it
is located.
As the final step, the thread atomically flips the \emph{result} pair from
some valid (previously used) pointer and era to \textbf{invptr}
and the current slow-path cycle tag obtained from the \emph{reservations} array.
Right after that instant, concurrent threads in \emph{increment\_era()}
(Line~96), called from \emph{alloc\_block()} and \emph{retire()}, can detect
that this thread needs helping.

While waiting for help, the thread resumes the loop where it attempts
to retrieve the hazardous reference. If the thread succeeds, it simply
cancels the request by flipping its \emph{state} to some valid pointer and era
(\textbf{nullptr} and $\infty$ in the pseudo-code). At that point, only
the tag from the \emph{reservations} array needs to be incremented to prepare
for the next slow-path cycle (Line~40). The corresponding WCAS (Line~38)
call can also fail, indicating that some helping thread already
produced the output which must be used instead. Either way,
the thread attempts to update its current reservation
by using WCAS (Line~45) which can only fail if a helping thread already
produced an output and modified the corresponding entry in \emph{reservations}.
The loop exits when the pointer field is no longer \textbf{invptr}.
Finally, the thread
modifies the entry in \emph{reservations} from the output value in the
\emph{result} pair. (Note that the entry can already be set by the helping thread,
in which case the same value will be just written again.)

We will now discuss the \emph{help\_thread()} procedure. First, this
procedure must set a reservation for the parent block, so that it can
be safely accessed (Line~107). While doing so, the original \emph{get\_protected()} may already complete. Thus, we check the tag (Line~110) before proceeding.
Once the reservation is set, the new \emph{cleanup()} routine will make sure that the parent block is not reclaimed (see Section~\ref{sec:correctness}). 
When \emph{help\_thread()}'s loop converges and the output is produced (Line~118), WCAS in Line~123 will attempt to update the \emph{reservations} array on
behalf of the \emph{get\_protected()} thread.
It only succeeds if the tag still refers to the previous slow-path cycle.

\section{Correctness}

\label{sec:correctness}

We assume that programs are \textit{well-behaved}, i.e., they call
provided API functions appropriately.
We focus on the arguments related to wait-freedom
and reclamation safety. General non-blocking
and memory usage arguments remain the same as in~\cite{pedroHEFULL}
since WFE is based on Hazard Eras.
We will denote the number of threads by $n$.

\begin{lemma}
The loop in Lines~34-49 is bounded by at most $n$ iterations.
\label{theorem:prop1}
\end{lemma}

\begin{proof}
A thread initiates the slow path in Line~33. The loop can only become
unbounded due to changing \emph{global\_era}. At most $n$ in-flight threads
can already be executing \emph{increment\_era()}, from \emph{alloc\_block()}
or \emph{retire()}, prior to Line~99, which
increments \emph{global\_era}, but after Line~96, which detects threads
that need helping.
Each of these in-flight threads will execute Line~99, potentially causing
the loop in \emph{get\_protected()} to fail and repeat.
Subsequent \emph{increment\_era()} calls detect threads that need helping
and only increment \emph{global\_era} after \emph{help\_thread()} (Line~97) is complete.
\end{proof}

\begin{lemma}
The loop in Lines~113-130 is bounded by at most $n$ iterations.
\label{theorem:prop2}
\end{lemma}

\begin{proof}
The proof is similar to that of Lemma~\ref{theorem:prop1}. The only
difference is that we need to also consider a case when the loop
potentially never terminates because the same thread keeps requesting the slow
path over and over again. However, Line~130 also checks the \emph{tag}
which guarantees that \emph{help\_thread} handles just one slow-path cycle.
\end{proof}

\begin{lemma}
The loop in Lines~120-125 is bounded by at most $2$ iterations.
\label{theorem:prop3}
\end{lemma}

\begin{proof}
For this loop to continue, the tag must remain intact (Line~122), i.e.,
the corresponding \emph{get\_protected()} call is still pending.
The loop is only executed when WCAS in Line~118 succeeds. Consequently,
the output is produced, which prompts Line~49 to terminate the slow-path
loop in \emph{get\_protected()}. Prior to that loop termination,
WCAS can still change the \emph{reservations} array one more time (Line~45).
Thus, the loop in Lines~120-125 may repeat one more time.
\end{proof}

\begin{lemma}
A parent object is not reclaimed while running help\_thread()
as long as the object reclamation procedure first checks
normal reservations [0..max\_hes-1] and then the first special reservation.
\label{theorem:prop4}
\end{lemma}

\begin{proof}
For an object to be removed, no reservation should overlap with it.
By the API convention, the parent object already has a corresponding
reservation for it (except when it is \textbf{nullptr}). 
However, this reservation is not guaranteed to last after
\emph{get\_protected()} is complete. If the tag in Line~110 matches the
corresponding field in \emph{result}, the reservation was set
while \emph{get\_protected()} was still active. For other cases,
we simply exit from \emph{help\_thread()}.
Since the order of making reservations coincides with
the order of checks, the parent object will be covered by at least
one reservation.
\end{proof}

\begin{lemma}
An object referred to by a hazardous entry is not reclaimed while
handing over a reservation from help\_thread() to get\_protected()
as long as the object reclamation procedure first checks the second
special reservation and then normal reservations [0..max\_hes-1].
\label{theorem:prop5}
\end{lemma}

\begin{proof}
Similar to a parent object, a hazardous entry object obtained
in \emph{help\_thread()} is protected by a special and normal reservations.
Compared to Lemma~\ref{theorem:prop4}, the order of
reservations is different. Using similar arguments as in
Lemma~\ref{theorem:prop4}, we conclude that the order of checks must
also happen in the opposite direction.
\end{proof}

\begin{theorem}
get\_protected() is wait-free bounded.
\end{theorem}

\begin{proof}
The number of iterations on the fast path is bounded. For the slow
path, Lemma~\ref{theorem:prop1} guarantees that the output must be produced
at most after $n$ iterations, or the corresponding loop must converge due to the global era value staying intact. When the output is produced, the slow-path loop terminates~(Line~49).
\end{proof}

\begin{theorem}
alloc\_block() is wait-free bounded.
\label{theorem:prop6}
\end{theorem}

\begin{proof}
\emph{alloc\_block()} periodically calls \emph{increment\_era()}.
Loops inside \emph{increment\_era()} are already bounded. On each
iteration, \emph{help\_thread()} is called. The \emph{help\_thread()} function
is bounded due to Lemmas~\ref{theorem:prop2} and \ref{theorem:prop3}.
\end{proof}

\begin{theorem}
retire() is wait-free bounded.
\end{theorem}

\begin{proof}
The proof is similar to that of Theorem~\ref{theorem:prop6}.
\end{proof}

\begin{theorem}
WFE's cleanup() is safe for memory reclamation.
\end{theorem}

\begin{proof}
The reservation scanning discipline in \emph{cleanup()} satisfies both Lemmas~\ref{theorem:prop4} and \ref{theorem:prop5}. It also satisfies Hazard Eras' original discipline for all other blocks.
\end{proof}

\section{Performance Results}

\begin{figure*}[ht]
\begin{subfigure}{.5\textwidth}
\centering
\includegraphics[width=.99\textwidth]{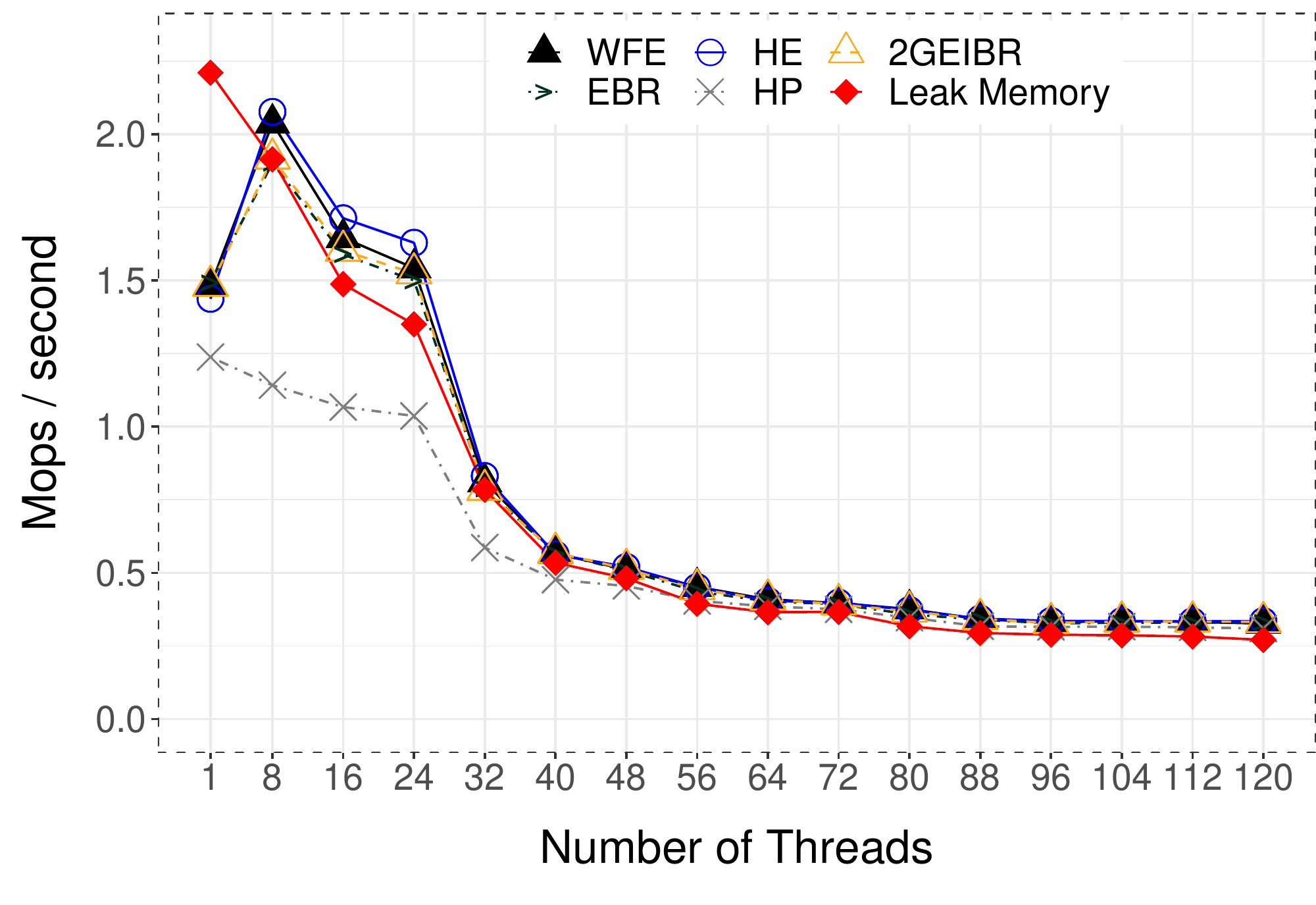}
\caption{KP (throughput)}
\label{fig:kp_thru}
\end{subfigure}%
\begin{subfigure}{.5\textwidth}
\centering
\includegraphics[width=.99\textwidth]{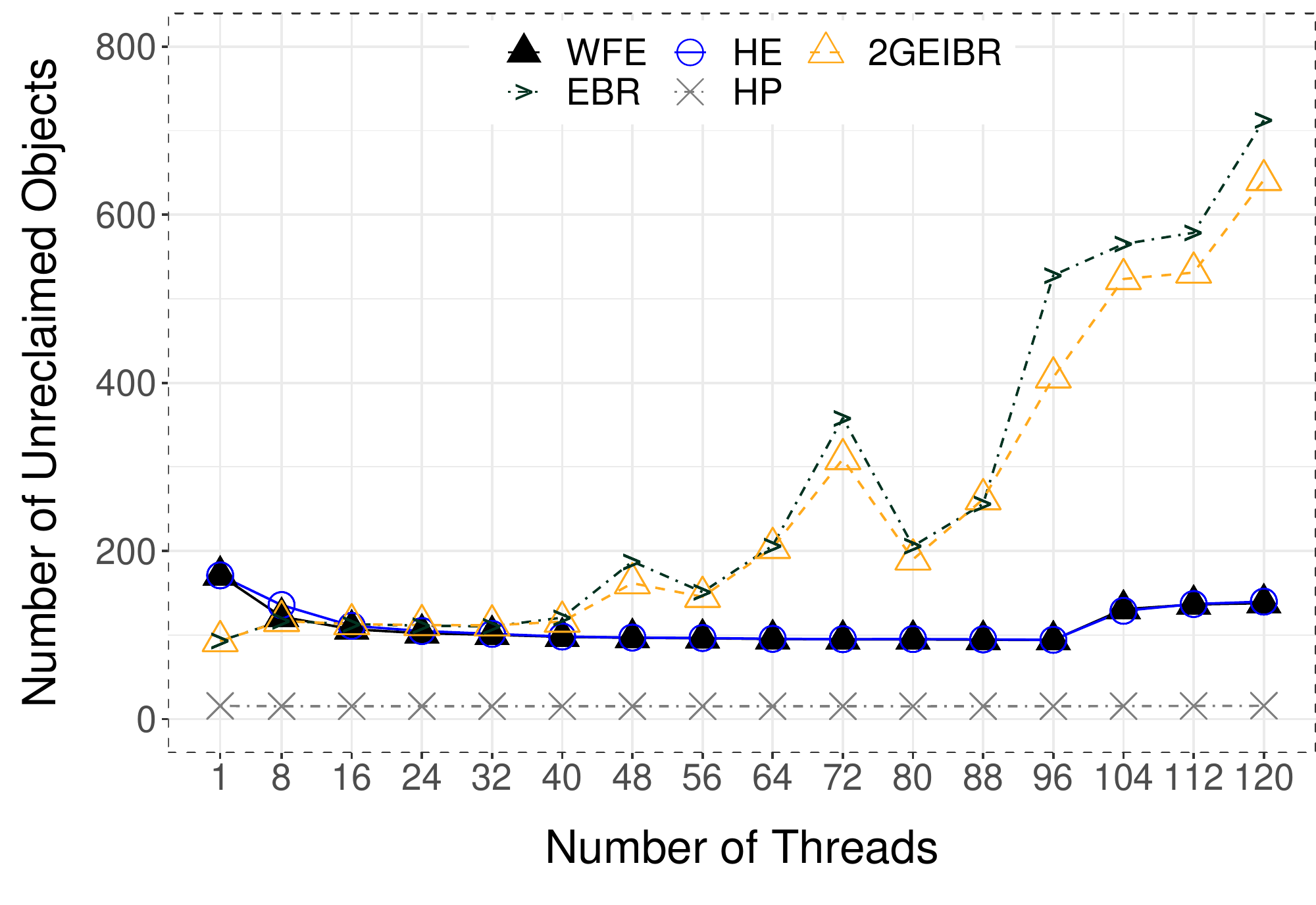}
\caption{KP (unreclaimed objects)}
\label{fig:kp_unrec}
\end{subfigure}%
\\
\begin{subfigure}{.5\textwidth}
\centering
\includegraphics[width=.99\textwidth]{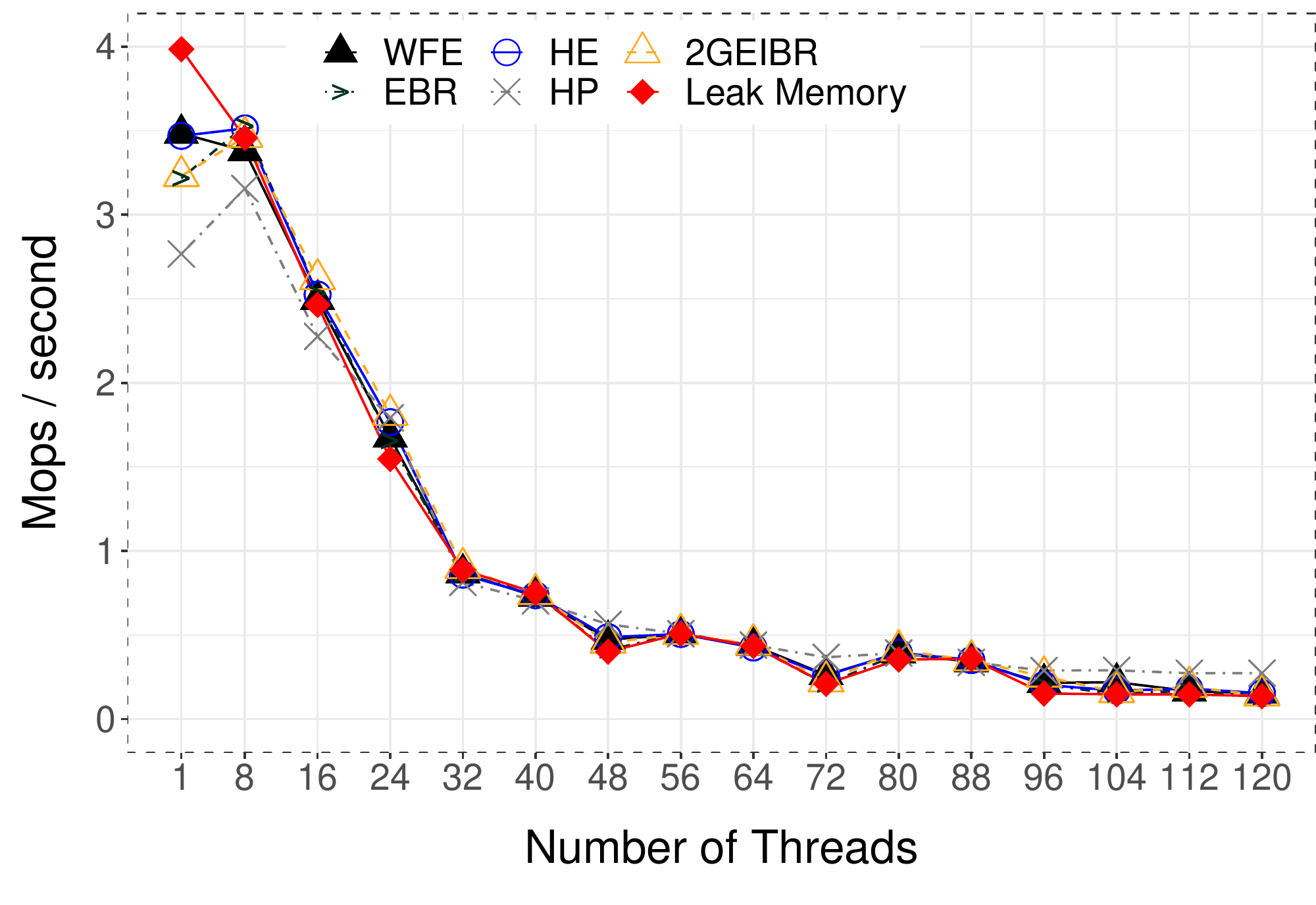}
\caption{CRTurn (throughput)}
\label{fig:crturn_thru}
\end{subfigure}%
\begin{subfigure}{.5\textwidth}
\centering
\includegraphics[width=.99\textwidth]{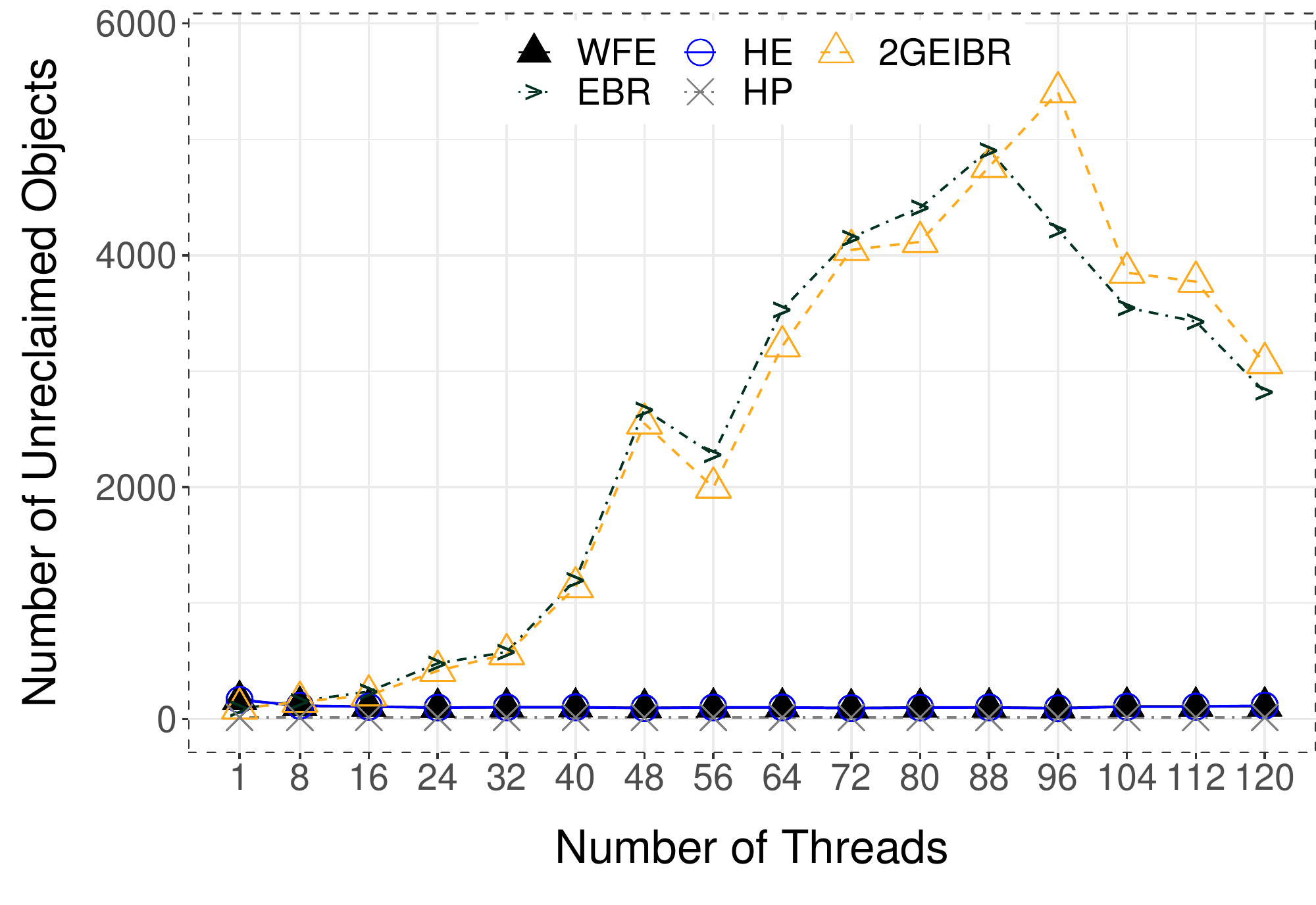}
\caption{CRTurn (unreclaimed objects)}
\label{fig:crturn_unrec}
\end{subfigure}%
\caption{Wait-free queues (\emph{50\% insert()} and \emph{50\% delete()}).}
\label{fig:wfqueue}
\end{figure*}

We performed all tests on an x86\_64 machine with 256GB of RAM and
four Intel Xeon~E7-8890~v4 (2.20GHz) processors, each with 24 cores. Processors have separate L1/L2 caches per each core and the L3 cache is shared across each processor. We pinned the first 24 threads to one processor, next 24 threads to another processor, and so on. We also disabled SMT (simultaneous multithreading), i.e., splitting one physical core into several  virtual cores, as it is typically recommended to disable SMT for more predictable measurements.

We ran the benchmark in~\cite{IBRPaper} that already implements
existing reclamation approaches and additionally
implemented our approach. The schemes include:
\begin{description}[leftmargin=.4cm]
\item[\textbf{WFE}:] our wait-free eras scheme presented in this paper.
\item[\textbf{HE}:] the hazard eras scheme~\cite{HEPaper}, which WFE extends.
\item[\textbf{HP}:] the classical hazard pointers scheme~\cite{HPPaper}.
\item[\textbf{EBR}:] the epoch-based reclamation scheme.
\item[\textbf{2GEIBR}:] the interval-based reclamation approach of~\cite{IBRPaper}; we used the 2GEIBR version which does not tag pointers.
\item[\textbf{Leak Memory}:] a baseline which just leaks memory, i.e., provides no memory reclamation.
\end{description}
We have fixed a potential race condition in HE's \emph{retire()}, which probably
went unnoticed in the original benchmark due to subtle differences in
\emph{retire()} between HE and IBR.

We compiled the benchmark, which is written in C++, using
g++ 8.3.0 (-O3 optimizations). Similar to~\cite{IBRPaper}, we used 
jemalloc~\cite{jemalloc} due to its better performance.

The goal of our evaluation is twofold. First, we wish to understand how universal our scheme is and its performance when used with common lock-free data structures. Second, we wish to understand the effectiveness of using our reclamation scheme with wait-free data structures while providing wait-free progress guarantees.
Since wait-free data structures are typically much harder to implement, we focused on a select few.

For lock-free data structures, we used the existing tests from the benchmark: a sorted \textit{Linked List}~\cite{HarrisList} (includes a modification from~\cite{HPPaper}), \textit{Natarajan BST} (binary search tree)~\cite{NatarajanTree}, 
and \textit{Hash Map}~\cite{HPPaper}.
For wait-free data structures, we extended the benchmark to implement Kogan-Petrank (\textit{KP})~\cite{kpWFQUEUE} and \textit{CRTurn}~\cite{pedroWFQUEUE} wait-free queues for all reclamation schemes; we based our implementation on the
existent code for Hazard Pointers~\cite{concurrencyFreaks}.
The original KP queue uses a garbage collector, for which
no known wait-free implementation exists; thus, we are the first
to evaluate the KP queue with wait-free reclamation.

In the evaluation, we used both write-dominated tests, where one half of all operations are insertions and the other half are deletions, as well as read-mostly tests, where $90\%$ of all operations are \emph{get()} and
the remaining $10\%$ are \emph{put()}.
Each data structure implements an abstract key-value interface
with the corresponding operations, i.e., \emph{insert()}, \emph{delete()},
\emph{get()}, and \emph{put()}.
Following the methodology from~\cite{IBRPaper}, each test measured a single data point by
pre-filling the data structure with 50K elements and then ran 10 seconds
(repeated 5 times). The key for each operation is randomly chosen
from the range $(0, 100000)$.

The benchmark allows tuning of certain parameters for each test. Specifically, the epoch counter for WFE, HE, EBR, and 2GEIBR is incremented after $n\times \nu$, where $n$ is the number of all active threads and
$\nu$ is the per-thread frequency of epoch or era increments. Similar to~\cite{IBRPaper}, we used $\nu=150$, which is large enough to avoid performance bottlenecks for the epoch counter increments in our setup. Similar to~\cite{IBRPaper}, the per-thread scanning frequency of retired lists is at least $30$ but depends on the algorithm due to differences in \emph{retire()}.
Finally, for WFE, we set the number of attempts on the fast-path to $16$.
Even if the number of attempts is that small, the slow path is taken
rarely. (We also tested our algorithm by forcing the slow path to
be taken all the time to validate that our implementation still works
correctly under stress conditions.)

\begin{figure*}[ht]
\begin{subfigure}{.5\textwidth}
\includegraphics[width=.99\textwidth]{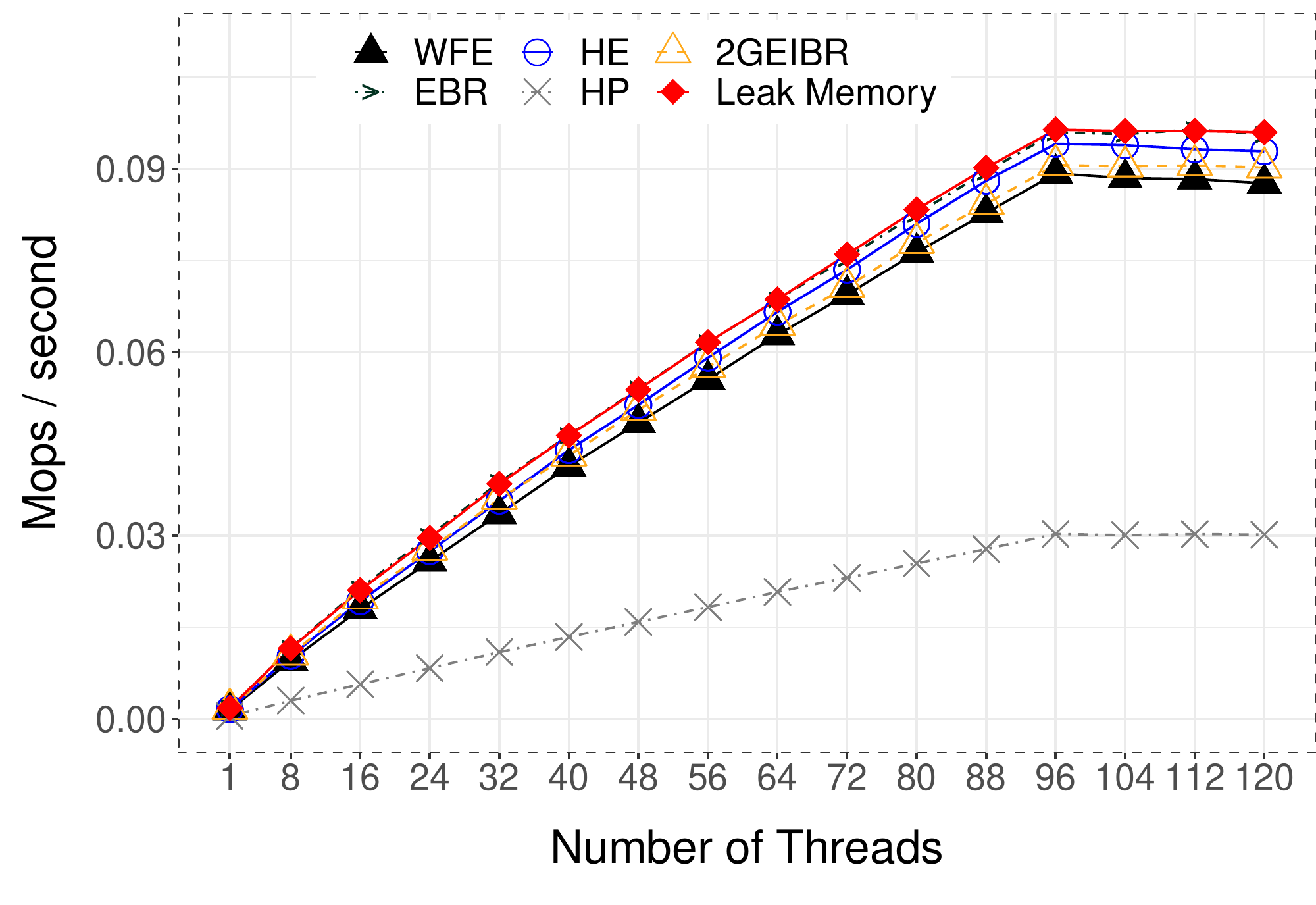}
\end{subfigure}%
\begin{subfigure}{.5\textwidth}
\includegraphics[width=.99\textwidth]{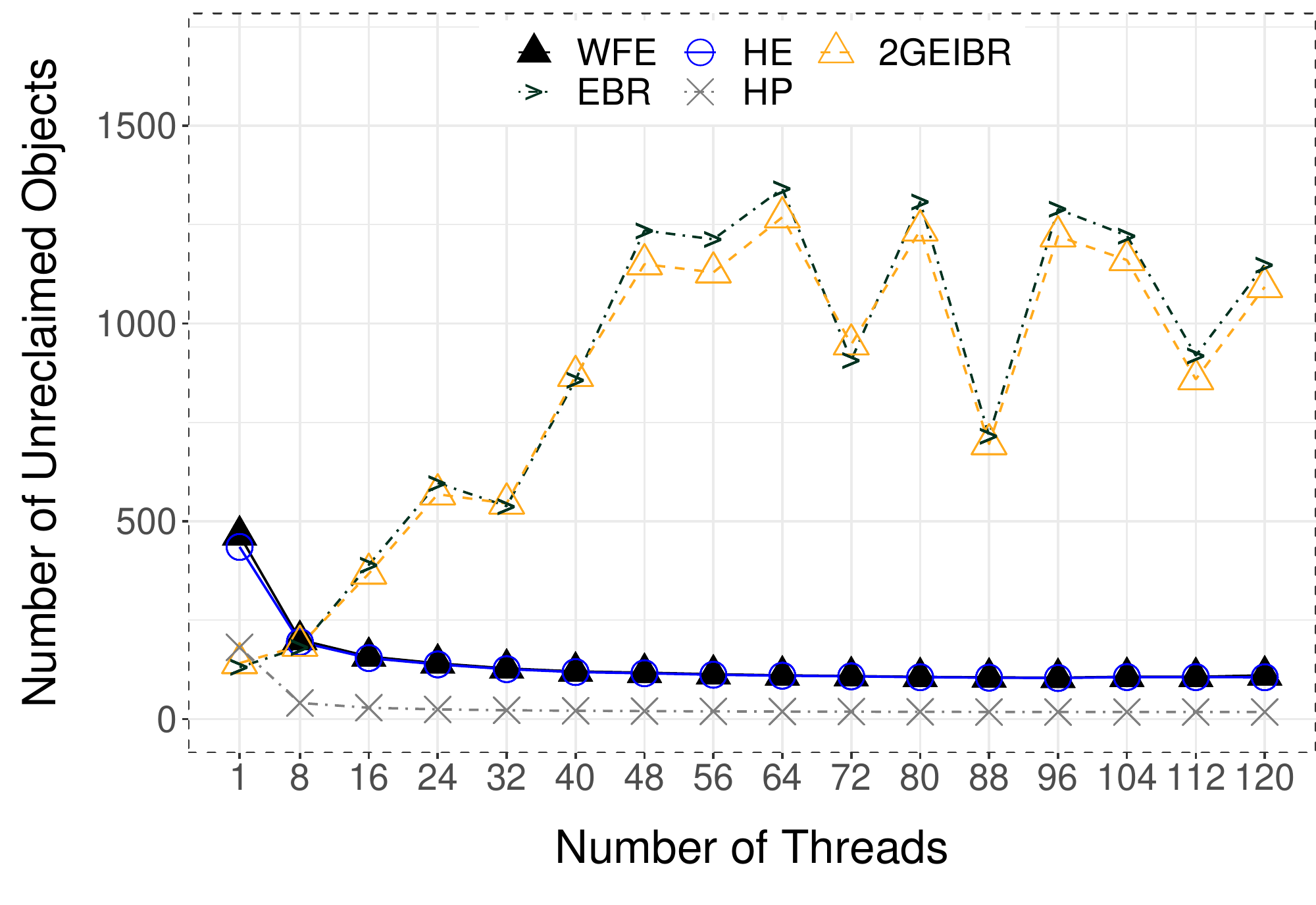}
\end{subfigure}%
\vspace{+2pt}
\caption{Linked List (\emph{50\% insert()} and \emph{50\% delete()}).}
\label{fig:list}
\end{figure*}

\begin{figure*}[ht]
\begin{subfigure}{.5\textwidth}
\includegraphics[width=.99\textwidth]{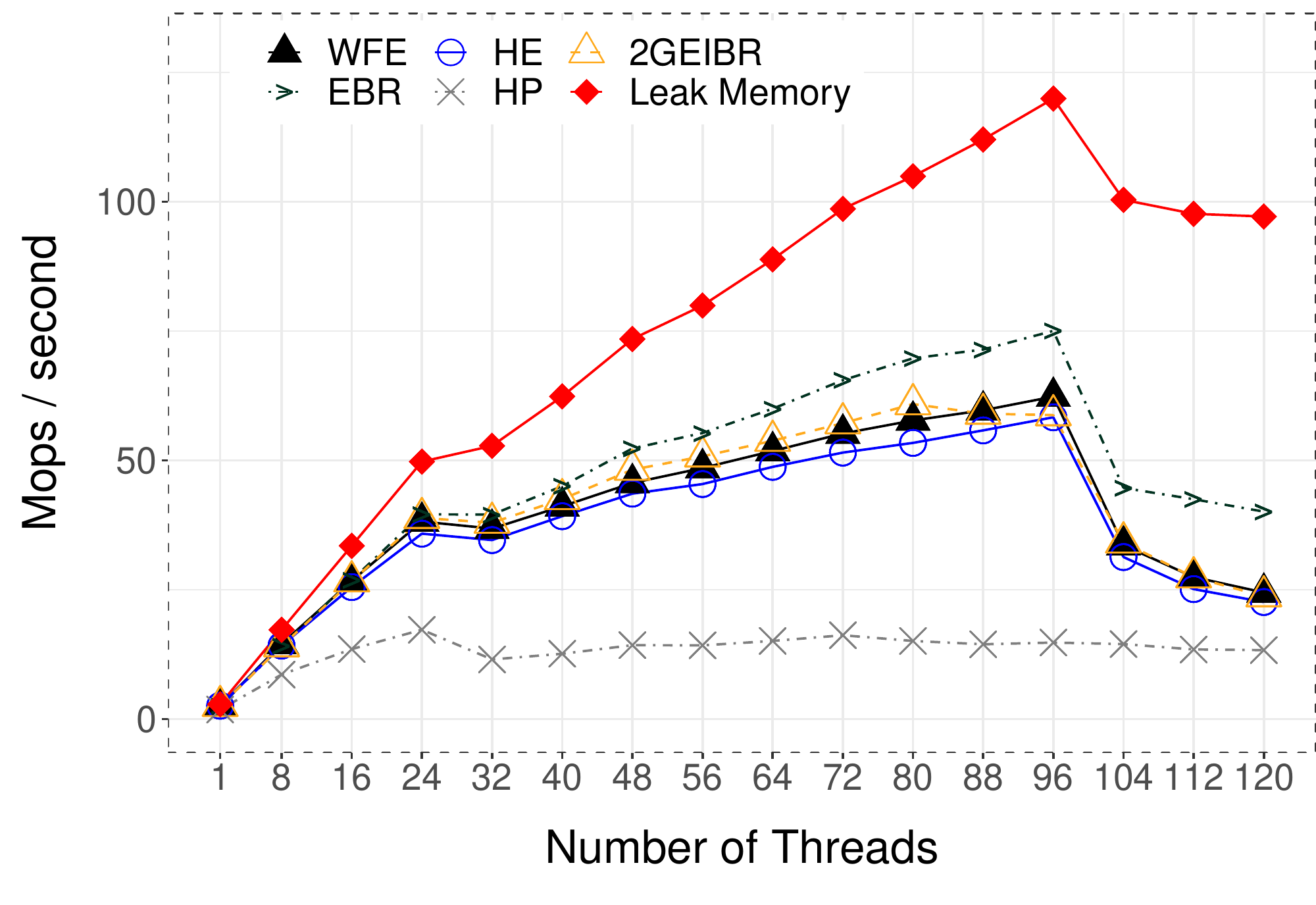}
\end{subfigure}%
\begin{subfigure}{.5\textwidth}
\includegraphics[width=.99\textwidth]{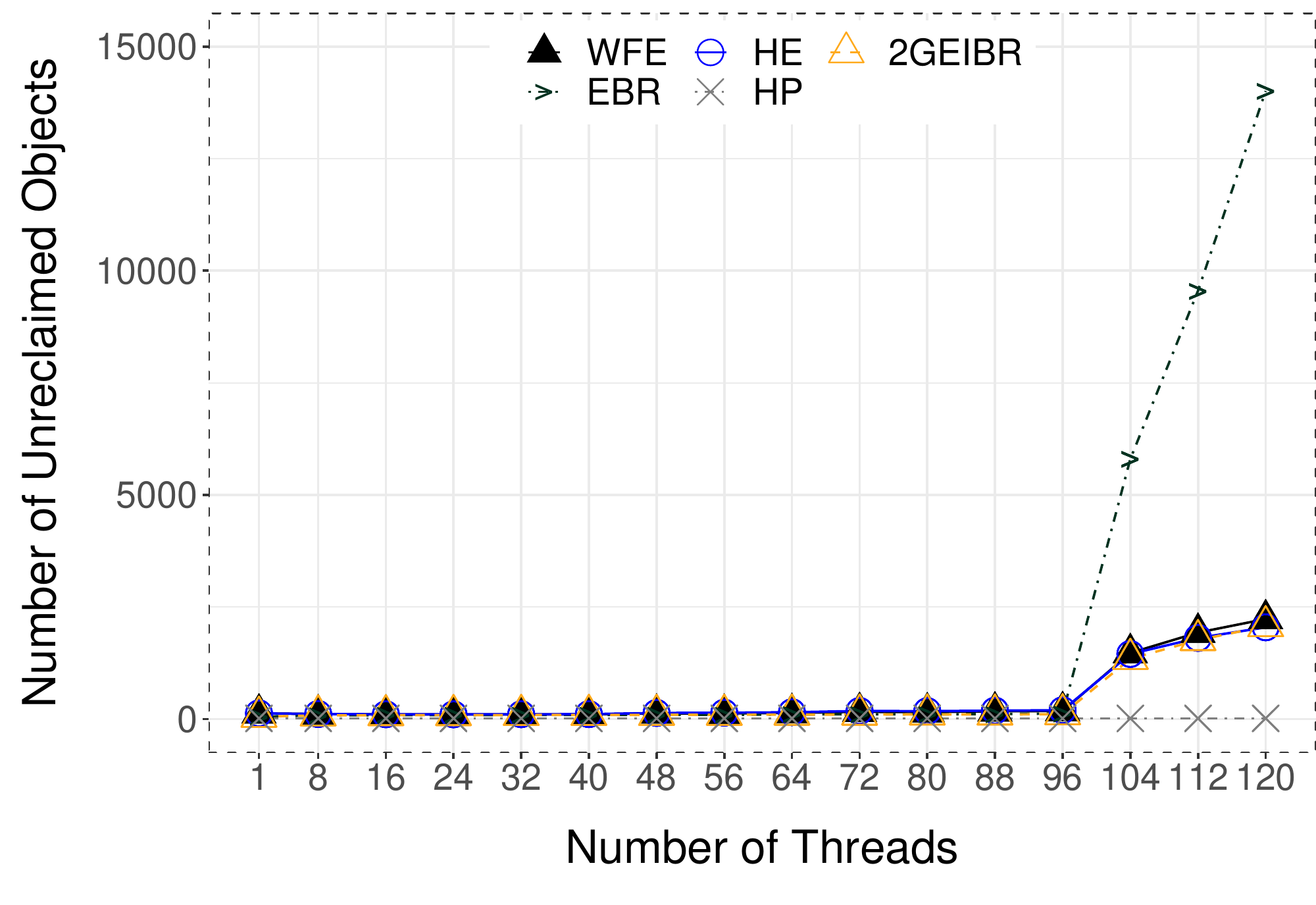}
\end{subfigure}%
\vspace{+2pt}
\caption{Hash Map (\emph{50\% insert()} and \emph{50\% delete()}).}
\label{fig:hash}
\end{figure*}

\begin{figure*}[ht]
\begin{subfigure}{.5\textwidth}
\includegraphics[width=.99\textwidth]{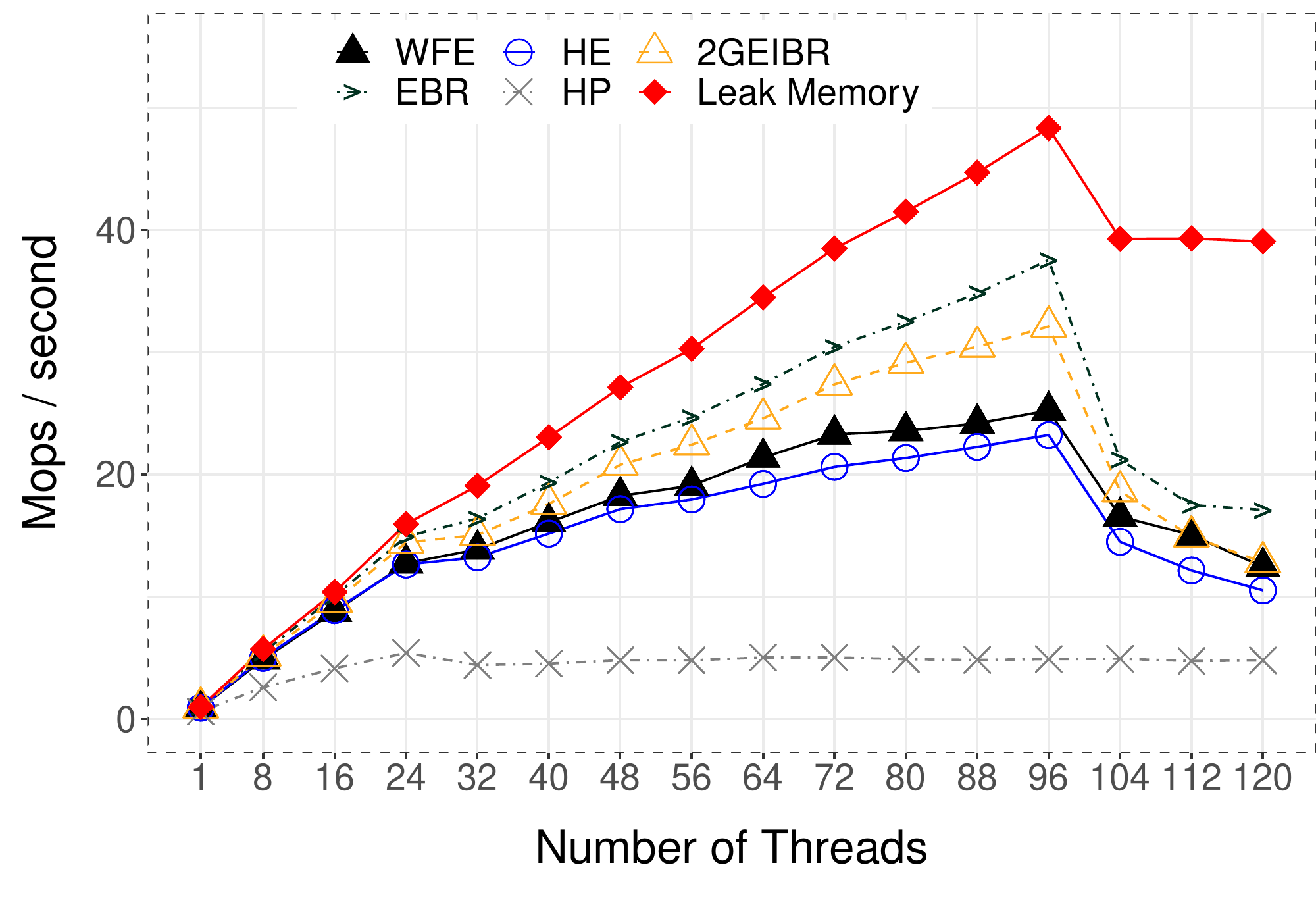}
\end{subfigure}%
\begin{subfigure}{.5\textwidth}
\includegraphics[width=.99\textwidth]{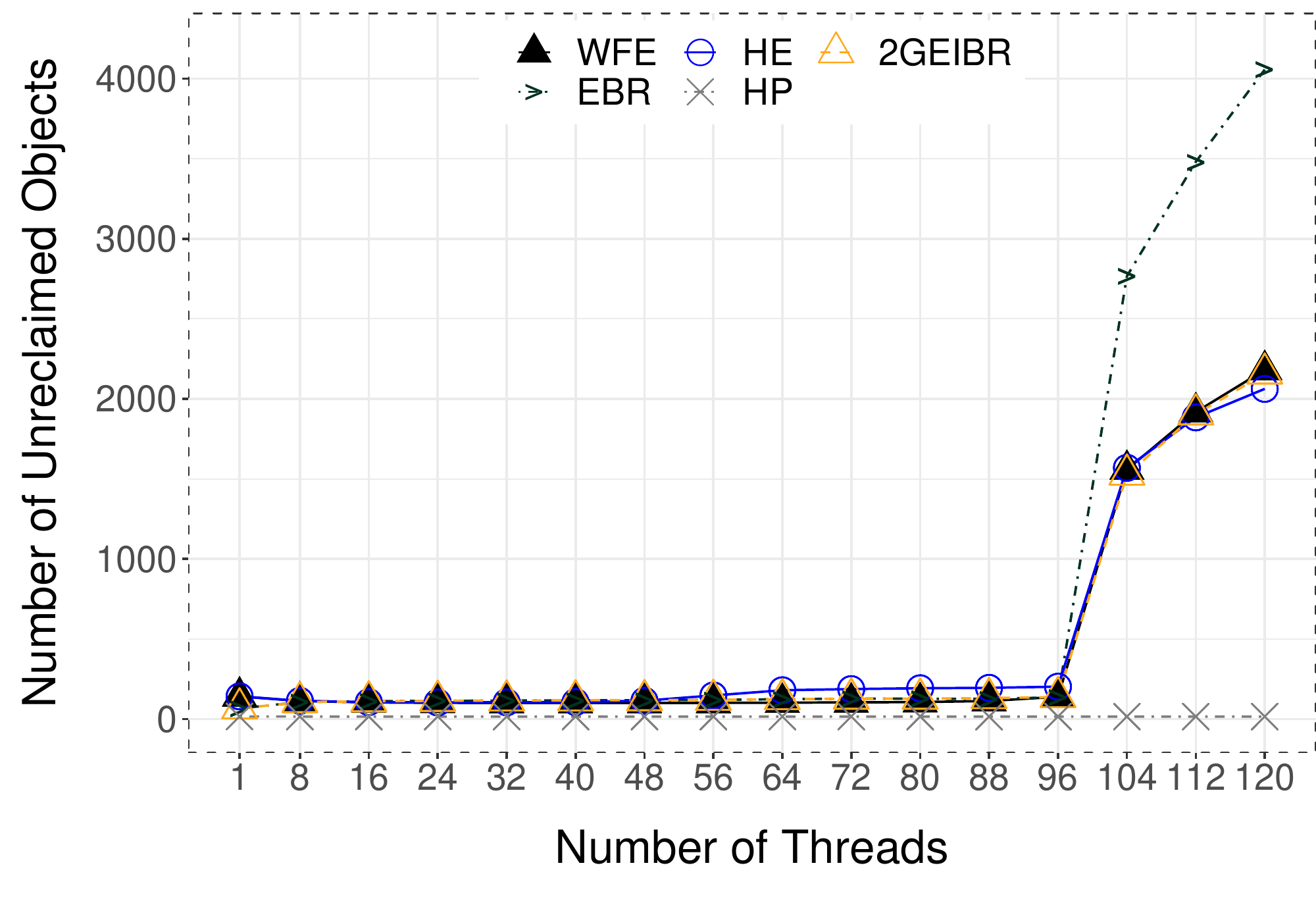}
\end{subfigure}%
\vspace{+2pt}
\caption{Natarajan BST (\emph{50\% insert()} and \emph{50\% delete()}).}
\label{fig:natarajan}
\end{figure*}

\begin{figure*}[ht]
\begin{subfigure}{.5\textwidth}
\includegraphics[width=.99\textwidth]{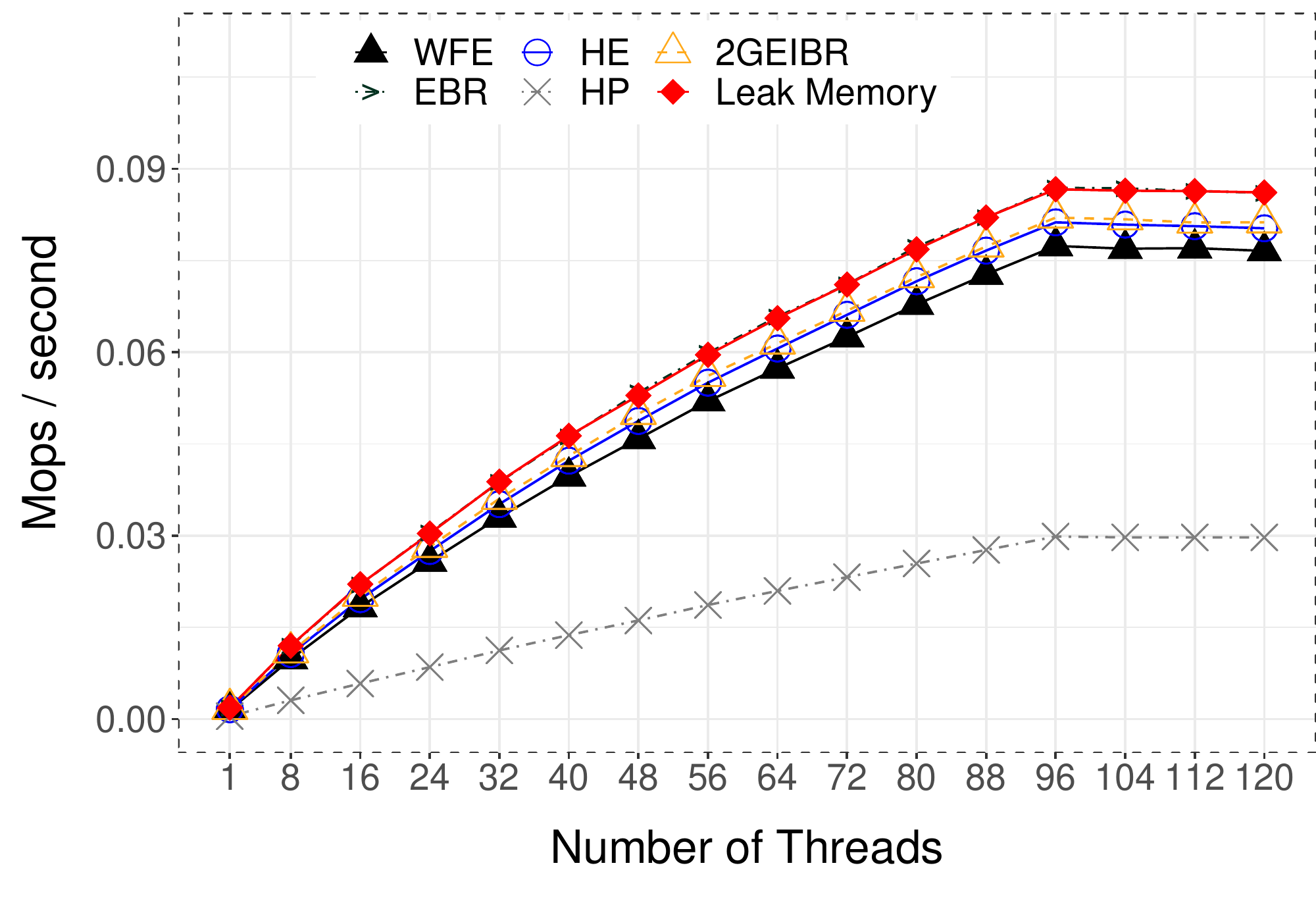}
\end{subfigure}%
\begin{subfigure}{.5\textwidth}
\includegraphics[width=.99\textwidth]{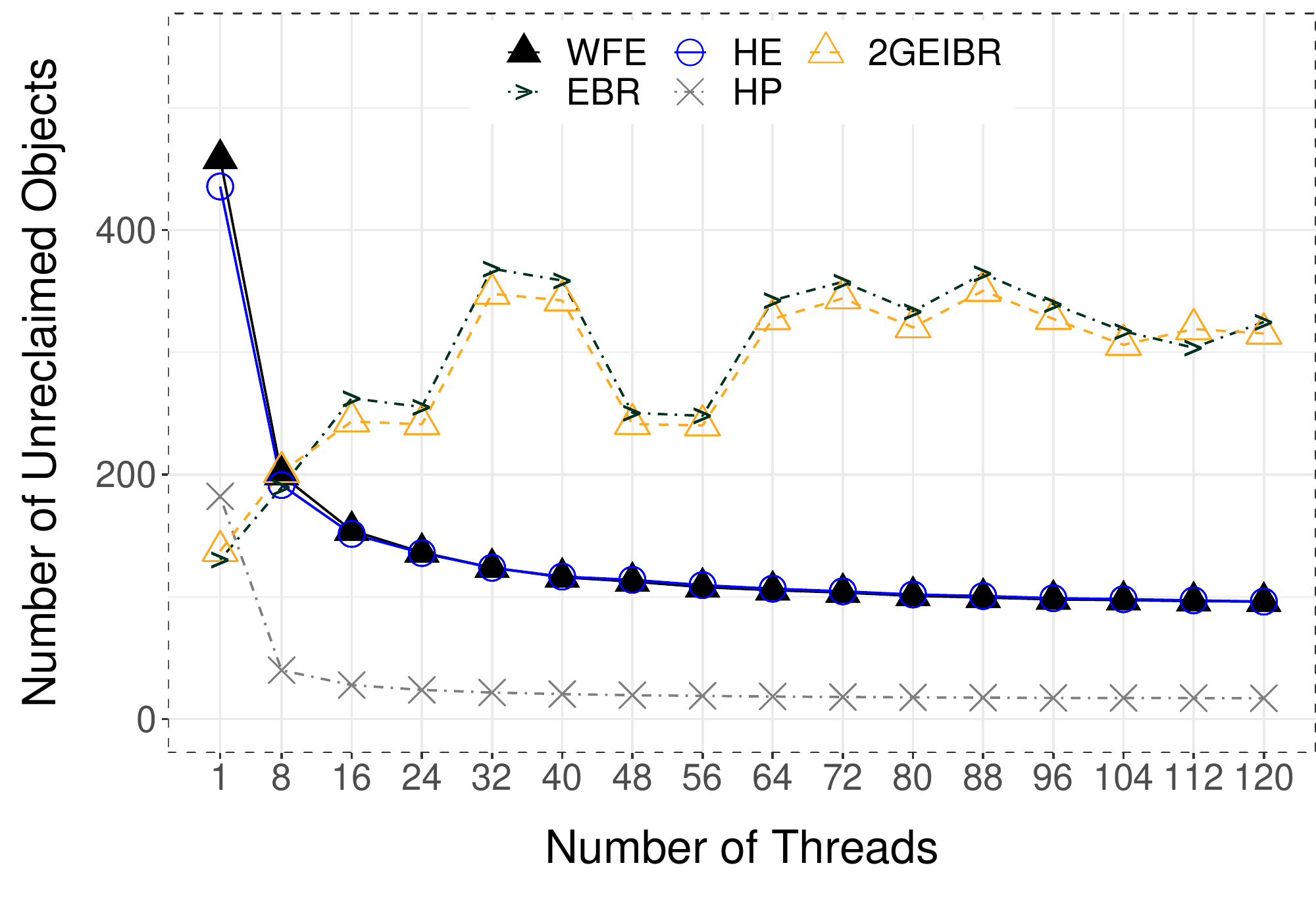}
\end{subfigure}%
\vspace{+2pt}
\caption{Linked List (\emph{90\% get()} and \emph{10\% put()}).}
\label{fig:list_read}
\end{figure*}

\begin{figure*}[ht]
\begin{subfigure}{.5\textwidth}
\includegraphics[width=.99\textwidth]{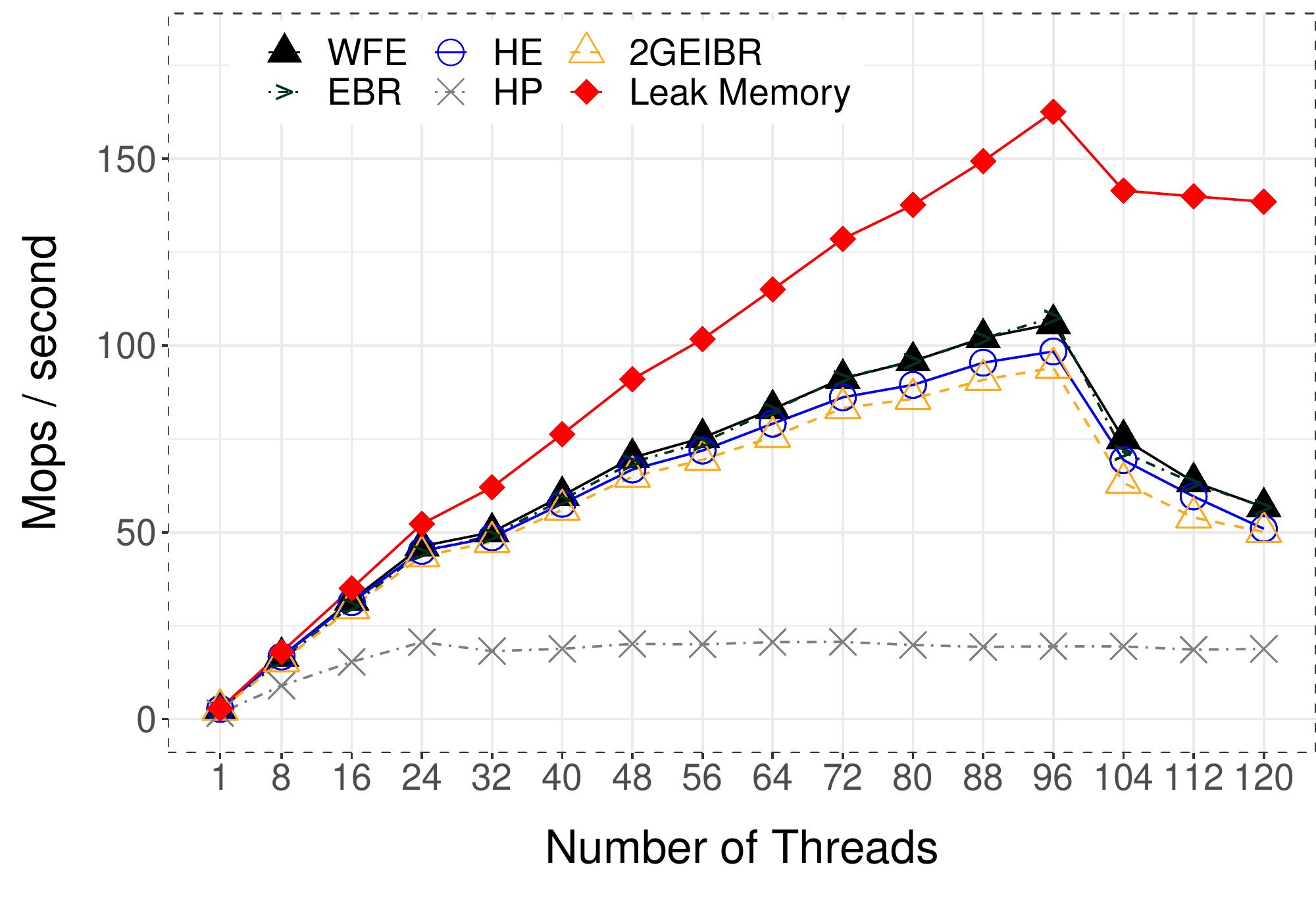}
\end{subfigure}%
\begin{subfigure}{.5\textwidth}
\includegraphics[width=.99\textwidth]{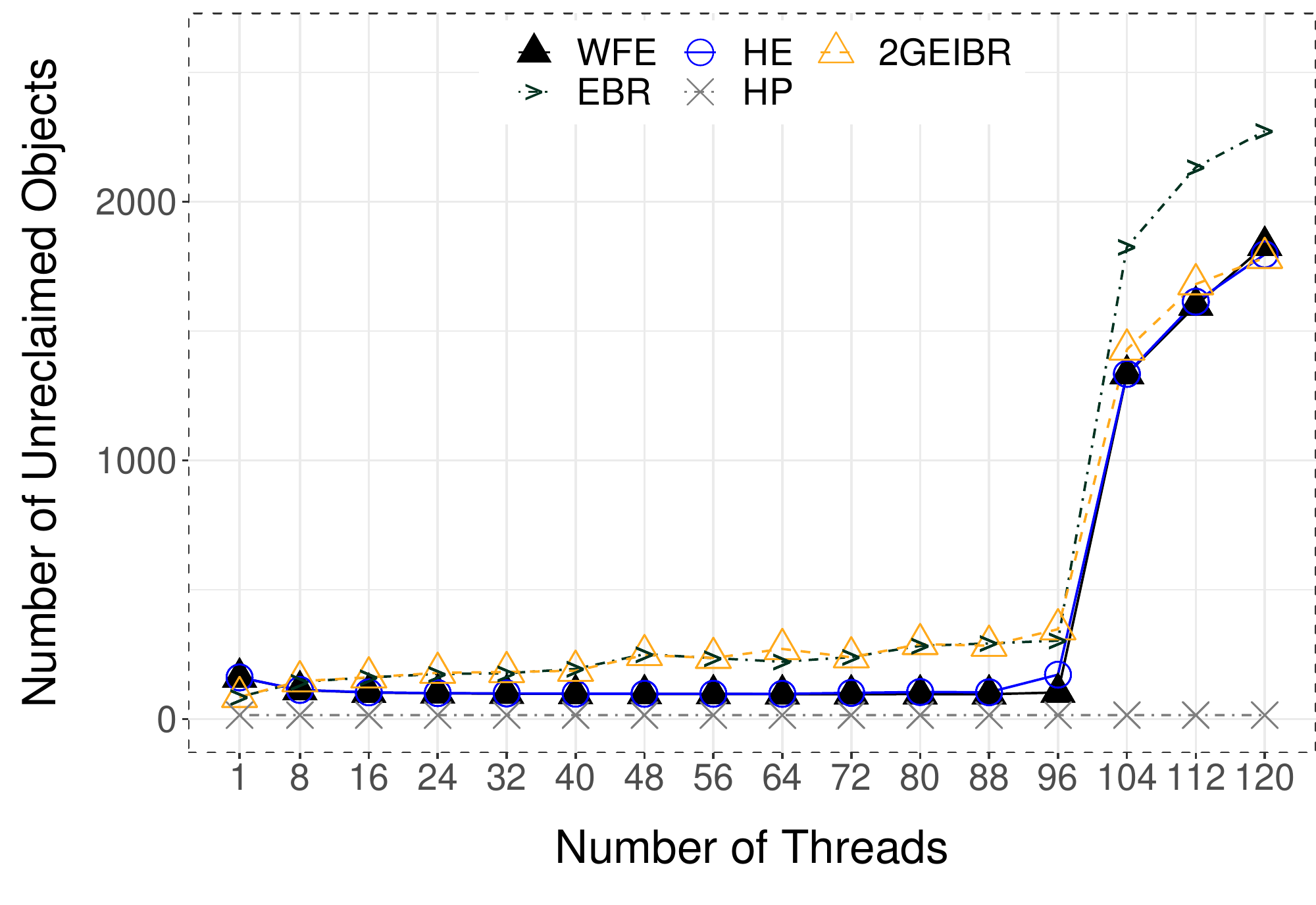}
\end{subfigure}%
\vspace{+2pt}
\caption{Hash Map (\emph{90\% get()} and \emph{10\% put()}).}
\label{fig:hash_read}
\end{figure*}

\begin{figure*}[ht]
\begin{subfigure}{.5\textwidth}
\includegraphics[width=.99\textwidth]{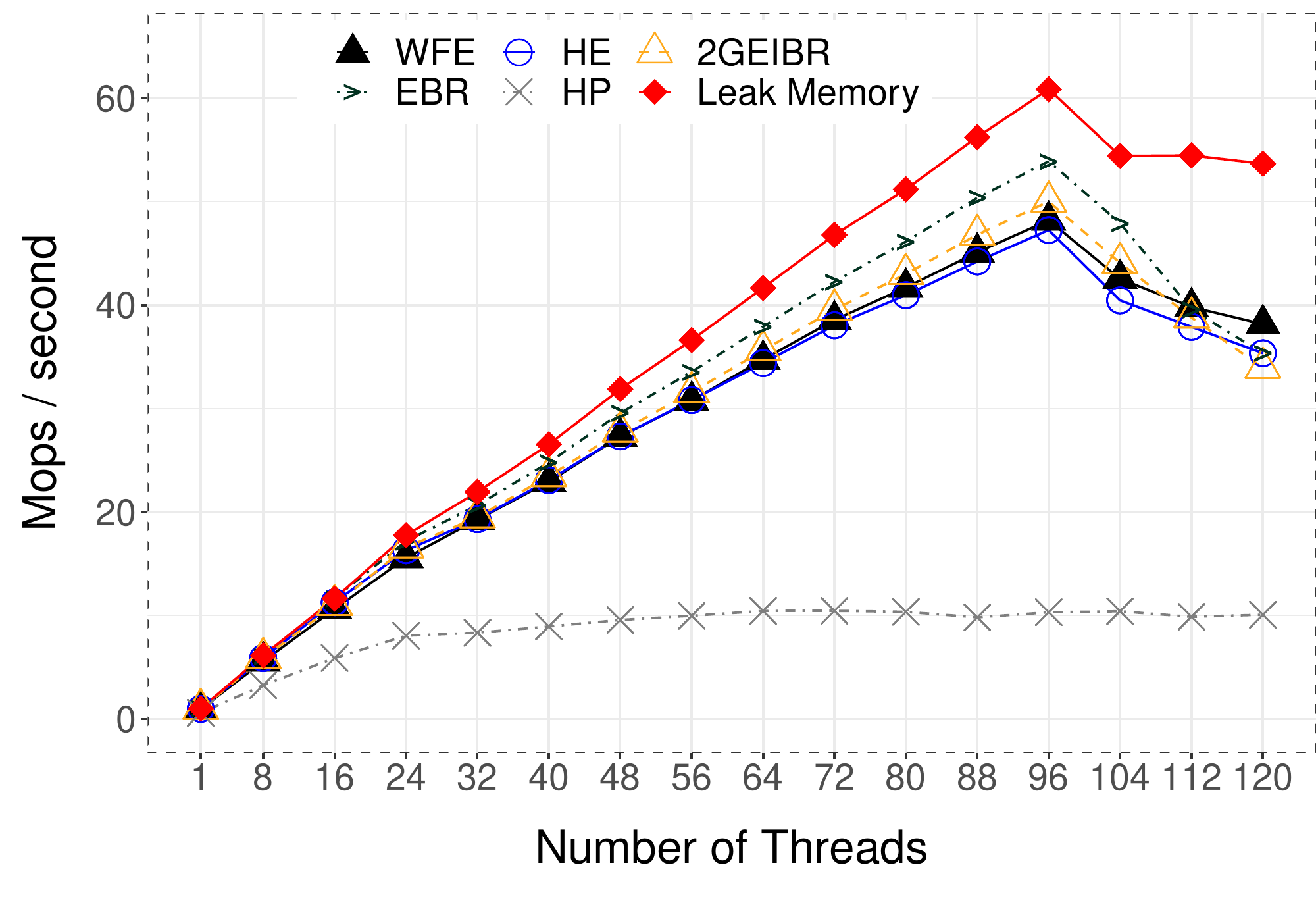}
\end{subfigure}%
\begin{subfigure}{.5\textwidth}
\includegraphics[width=.99\textwidth]{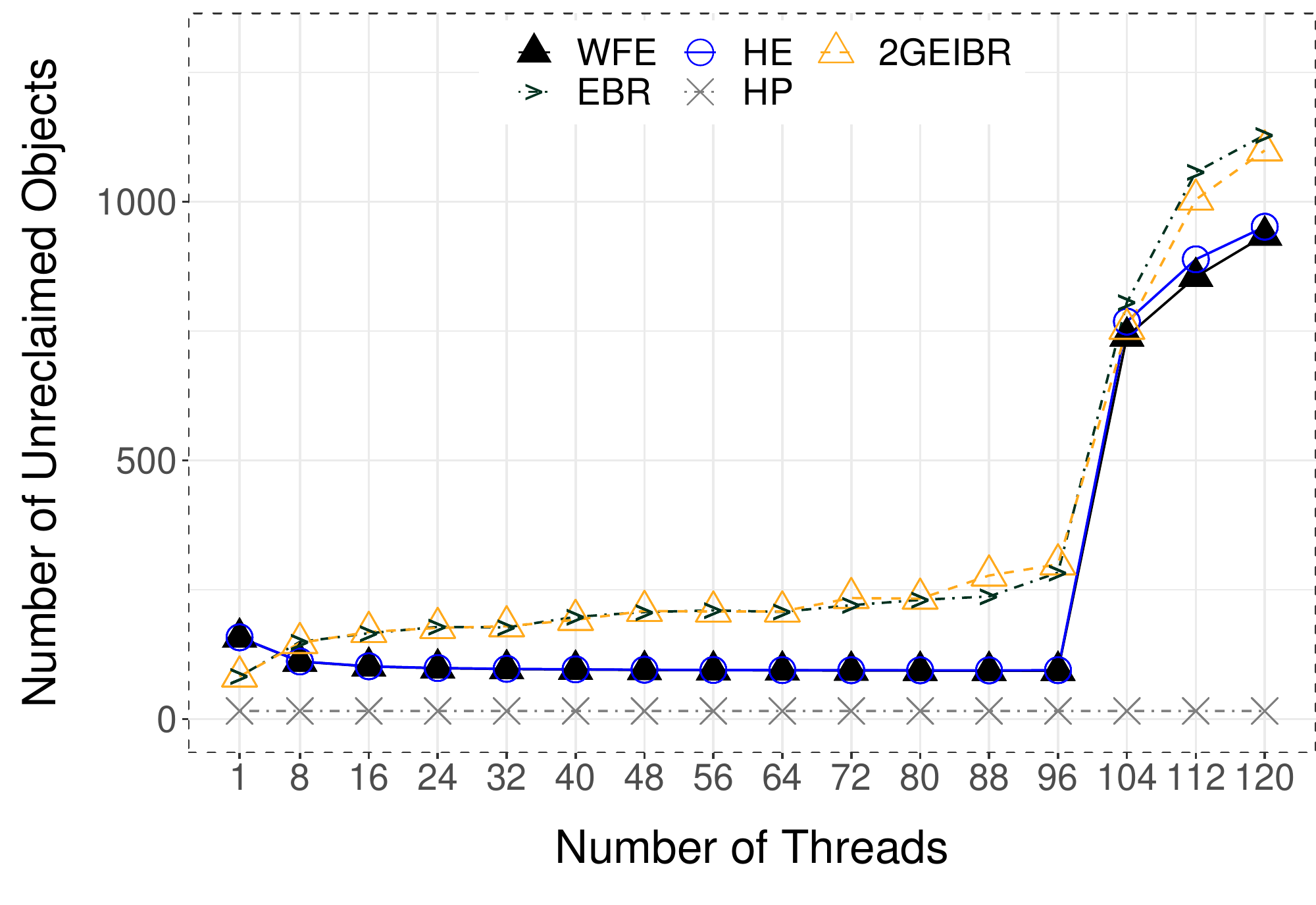}
\end{subfigure}%
\vspace{+2pt}
\caption{Natarajan BST (\emph{90\% get()} and \emph{10\% put()}).}
\label{fig:natarajan_read}
\end{figure*}

Figure~\ref{fig:wfqueue} shows the throughput for wait-free queues.
For queues, typically, only \emph{insert()} and \emph{delete()} operations
make sense. Thus, we only present results for the write-dominated workload.
Generally speaking, both KP and CRTurn queues show similar throughput (Figures~\ref{fig:kp_thru}~and~\ref{fig:crturn_thru}) for all schemes except HP, which is sometimes slower. Queues generally do not scale very well.
With respect to the average number of unreclaimed objects per
operation, a metric which measures the memory reclamation speed, we
found that HE and WFE are slightly less efficient than HP, but better than EBR and 2GEIBR (Figures~\ref{fig:kp_unrec}~and~\ref{fig:crturn_unrec}).

For Linked-List (Figures~\ref{fig:list}~and~\ref{fig:list_read}), we found that EBR is margina\-lly
better than all other schemes except HP, which exhibits the worst performance.
WFE is marginally worse than HE. Our investigation has shown that
an average linked-list traversal operation is long (due
to sequential search) and dereferences many pointers. Consequently,
(inlined) \emph{get\_protect\-ed()} calls have to be very efficient. Since
WFE also needs to call the slow-path procedure, higher register pressure forces the compiler to generate less efficient code.
However, this overhead is still quite insignificant. If desired, the
overhead can be eliminated by customizing a calling convention for
the slow-path procedure: the customized call can reduce the register pressure on the fast path by ensuring that no registers
need to be saved by callees.
With respect to the average number of unreclaimed objects per
operation, except for smaller concurrency, HE and WFE are less efficient than HP, but better than EBR and 2GEIBR.

For Hash-Map (Figures~\ref{fig:hash}~and~\ref{fig:hash_read}), we found that EBR is at the same level as or
marginally outperforms other schemes (except HP which has the worst
performance). WFE is at the same level as HE or even marginally outperforms it.
Natarajan BST (Figures \ref{fig:natarajan}~and~\ref{fig:natarajan_read})
shows similar trends, except that the gap between EBR and
HE (or WFE) is larger.
For the write-dominated tests, we found that EBR is
significantly less memory efficient than all other
schemes when threads are preempted.

Overall, the results show that WFE's performance is comparable to
that of other high-performant non-blocking algorithms such as HE and 2GEIBR.
At the same time, WFE provides the stronger wait-free progress guarantee.

\section{Related Work}
\label{sec:related}

The literature presents a number of memory reclamation techniques for concurrent data structures. We classify them into different categories.

The \textit{first category} includes schemes that use epochs such as EBR~\cite{epoch1}, which originates from RCU~\cite{Mckenney01read-copyupdate}. In these approaches, a thread records the global epoch value to make a reservation at the beginning of an operation. Then, at the end of the operation, it resets the reservation. A related approach, quiescent-state reclamation~\cite{epoch2}, increments the counter after all threads transition through a state where they hold no pointers. Stamp-it~\cite{Poter:2018:BAS:3210377.3210661} extends EBR to bound reclamation cost to $O(1)$.

Due to potentially unbounded memory usage, all these techniques can be blocking when threads are preempted or stalled. Hazard Eras~\cite{HEPaper} and IBR~\cite{IBRPaper} implement
a non-blocking epoch-based approach. Our Wait-Free Eras scheme extends
Hazard Eras, but the same idea can also be straightforwardly applied to certain versions of IBR, e.g., 2GEIBR.

The \textit{second category} includes reclamation techniques that deal with pointers. Hazard Pointers~\cite{HPPaper} record all pointers that are currently
in use. The technique has a relatively high overhead due to its extensive use of memory barriers for each pointer dereference. The original paper
presents Hazard Pointers as a ``wait-free'' scheme. However, the difficult
part comes during traversals, when an advertised pointer changes and needs to
be read again.
(Granted,~\cite{HPPaper} sidesteps an explicit \emph{get\_protected()}
operation, which we discuss in this paper.)
Pass-the-buck~\cite{Herlihy:2002:ROP:645959.676129,Herlihy:2005:NMM:1062247.1062249} uses a similar model.
Another technique, drop-the-anchor~\cite{Braginsky:2013:DAL:2486159.2486184}, is designed specifically for linked-lists, 
and outperforms hazard pointers. This approach, however, does not seem to be directly applicable to other data structures.
Optimistic Access~\cite{Cohen:2015:EMM:2755573.2755579} is more general, and leverages a ``dirty'' flag instead of
publishing hazard pointers, but requires data structures to be written in a ``normalized form.'' Automatic Optimistic Access~\cite{Cohen:2015:AMR:2814270.2814298} relies on a data structure-specific garbage collector
to make reclamation more automatic, but still requires data structures to be written in a normalized form.  FreeAccess~\cite{Cohen:2018:DSD:3288538.3276513} forgoes this requirement by extending the LLVM compiler
to make the process fully automatic.

We considered making Hazard Pointers wait-free. Just like Hazard Eras,
the Hazard Pointers scheme is also mostly wait-free except
the \emph{get\_protected()} operation. However, Hazard Pointers use pointers
instead of epochs, and there does not seem to be a straightforward way to adopt our approach for Hazard Pointers or any other technique that tracks pointers in the same manner.
 
The \textit{third category} is reference counting~\cite{refcount1,refcount2,refcount3,refcount4}.
A memory object is reclaimed when the reference counter, associated with
the object, reaches zero. General-purpose reference counting is typically
lock-free and very intrusive: a pointer
and a reference counter are adjacent and need to be both atomically updated
using WCAS. Reference counting typically performs poorly on
read-dominated workloads, as read operations must update reference counters,
which requires additional memory barriers.
Hyaline~\cite{Hyaline}, an
approach that implements distributed reference counting, forgoes this
requirement and achieves excellent performance.
However, Hyaline is still only lock-free.

Reclamation schemes of the \textit{fourth category} use special OS mechanisms.
For these reclamation schemes, it is generally hard to guarantee 
non-blocking behavior since an OS can use locks internally.
DEBRA+~\cite{DEBRAPaper} uses signals to add fault tolerance to EBR, i.e., a
stalled thread which does not advance its epoch is interrupted by an OS signal.
This signal triggers a restart operation, for which special recovery code
needs to be written.
ThreadScan~\cite{ThreadScan} implements a mechanism which uses a shared
\textit{delete buffer}. A thread triggers reclamation by sending signals
to all other
active threads, which scan their stacks and registers to mark deleted
nodes in the shared buffer if they are still used.
ForkScan~\cite{ForkScan} also uses signals as well as copy-on-write OS
optimizations for fork(2). A reclaimer thread creates a child process which
contains a ``frozen'' memory snapshot; this process can scan deleted nodes in
parallel.
QSense~\cite{Balmau:2016:FRM:2935764.2935790} uses quiescent-state reclamation in its fast path, but hazard pointers back it
up when threads do not respond. It also relies on the specific behavior of an OS scheduler, which needs to context switch threads periodically.
Yet another approach~\cite{10.1145/2926697.2926699} uses an OS page fault
mechanism to alleviate the costs related to memory barriers in hazard pointers.

Software transactional memory (STM) can simplify concurrent
programming and memory reclamation. OneFile~\cite{OneFile} is a recent
wait-free STM implementation with its own, STM-specific memory reclamation. While OneFile's framework enables the implementation of a wide range of wait-free algorithms by directly converting sequential data structures, customized wait-free data structures can often better
utilize parallelism and achieve higher overall performance. 
Our work, which presents a universal memory reclamation scheme for such
\textit{arbitrary} data structures, closes this gap.

Hardware transactional memory (HTM) mechanisms are also widely used
in concurrent programming.
For example, reference counting can be accelerated using HTM~\cite{HTMRefcntPaper}.
Another approach, StackTrack~\cite{StackTrack}, encapsulates read operations in HTM transactions;
a concurrent thread will abort an HTM transaction when an object is no longer live but still in use.
Since typical HTMs lack wait-free progress guarantees, it is not clear how to use HTMs for wait-free memory reclamation.

\section{Conclusion}

We presented the first practical wait-free memory
reclamation scheme called \emph{Wait-Free Eras} (WFE). Like Hazard Eras, it achieves great performance while providing compatibility with Hazard Pointers. Unlike Hazard Eras or any other existing universal technique, WFE guarantees
that \textit{all} memory reclamation operations are wait-free. 

Since WFE eliminates a serious obstacle in wait-free programming, i.e., wait-free memory reclamation, we hope that it spurs further research in this area and the practical adoption of wait-free algorithms.

\section*{Availability}
WFE's code is available at \url{https://github.com/rusnikola/wfe}.

\section*{Acknowledgments}
We thank the anonymous reviewers for their insightful comments, which helped improve the paper. This work is supported in part by ONR under grants N00014-16-1-2711 and N00014-18-1-2022 and AFOSR under grant FA9550-16-1-0371.

\bibliography{lockfree}

\end{document}